\documentclass[10pt]{article}
\usepackage{hyperref}
\usepackage{bbm}
\usepackage{amssymb,amsmath,amsthm, fullpage}
\usepackage{color}
\usepackage{tikz}
\usepackage[mathscr]{euscript}
\usepackage{enumerate}
\usepackage[conditional,light,first,bottomafter]{draftcopy}
\draftcopyName{DRAFT\space\today}{130}
\draftcopySetScale{65}
\usepackage[letterpaper,hmargin=3.3cm,vmargin={2.6cm,3.2cm}]{geometry}
\usepackage{setspace}
\makeatletter
\renewcommand{\section}{\@startsection%
{section}%
{1}%
{0em}%
{1.7em}%
{1.2em}%
{\normalfont\large\centering\bfseries}}
\renewcommand{\@seccntformat}[1]%
{\csname the#1\endcsname.\hspace{0.5em}}
\makeatother
\renewcommand{\thesection}{\arabic{section}}
\numberwithin{equation}{section}
\renewcommand\appendix{\par
\setcounter{section}{0}%
\setcounter{subsection}{0}%
\setcounter{theorem}{0}
\setcounter{table}{0}
\setcounter{figure}{0}
\gdef\thetable{\Alph{table}}
\gdef\thefigure{\Alph{figure}}
\section*{Appendix}
\gdef\thesection{\Alph{section}}
\setcounter{section}{1}}
\newtheorem{theorem}{Theorem}[section]
\newtheorem{proposition}[theorem]{Proposition}
\newtheorem{lemma}[theorem]{Lemma}
\newtheorem{corollary}[theorem]{Corollary}
\theoremstyle{definition}
\newtheorem{definition}{Definition}
\newtheorem{remark}{Remark}
\newtheorem{notation}{Notation}

\newcommand{\reals}{\mathbb{R}}
\newcommand{\nats}{\mathbb{N}}

\newcommand{\complex}{\mathbb{C}}
\newcommand{\abs}[1]{\left|#1\right|}
\newcommand{\norm}[1]{\left\|#1\right\|}
\newcommand{\I}{{\rm i}}
\newcommand{\inner}[2]{\left\langle#1,#2\right\rangle}

\newcommand{\cH}{{\cal H}}
\newcommand{\cc}[1]{\overline{#1}}
\newcommand{\eval}[1]{\upharpoonright_{#1}}
\newcommand{\convergesto}[2]{\xrightarrow[#1\to #2]{}}

\DeclareMathOperator{\im}{Im}
\DeclareMathOperator{\dom}{dom}

\DeclareMathOperator{\ran}{ran}

\DeclareMathOperator{\clos}{clos}
\DeclareMathOperator*{\slim}{s-lim}
\begin{document}
%\begin{titlepage}
\title{\sc Scattering theory for a class of non-selfadjoint  extensions of symmetric operators
%and applications to scattering theory
%%%\footnotetext{%
%%%Mathematics Subject Classification(2010):
%%%XXXXX,  % Inverse problems (34-XX Ordinary differential equations 34Kxx
       % Functional-differential and differential-difference
       % equations)
%%%XXXXX, % (47-XX Operator theory 47Axx General
       % theory of linear operators)
\footnotetext{%
Mathematics Subject Classification (2010):
%47A10 % Spectrum, resolvent
%47A20 % Dilations, extensions, compressions
47A45 % Canonical models for contractions and nonselfadjoint operators
%47A55 % Perturbation theory
%47B25 % Symmetric and selfadjoint operators (unbounded)
%47B44 % Accretive operators, dissipative operators, etc.
34L25 % Scattering theory, inverse scattering
%81U40 % Inverse scattering problems
81Q35 % Quantum mechanics on special spaces: manifolds, fractals, graphs, etc.
%34A55 % Inverse problems (for ordinary differential equations)
}
\footnotetext{%
Keywords:
Functional model;
Extensions of symmetric operators; Boundary triples;
Scattering theory
%Quantum Graphs;
}
}
%\footnotetext{%
%Keywords:
%Functional model;
%Extensions of symmetric operators;
%Smooth vectors
%}
%}
\author{
\textbf{Kirill D. Cherednichenko}
\\
%% ----- Institution --------
\small Department of Mathematical Sciences,\\[-1.1mm]
\small University of Bath,\\[-1.1mm]
\small
Claverton Down, Bath BA2 7AY, U.K.\\[-1.1mm]
\small
\texttt{K.Cherednichenko@bath.ac.uk}
\\[2mm]
\textbf{Alexander V. Kiselev}
\\
%% ----- Institution --------
\small
Departamento de F\'{i}sica Matem\'{a}tica, 
\\[-1.1mm]
\small
Instituto de Investigaciones en Matem\'aticas Aplicadas y en Sistemas, 
\\[-1.1mm]
\small
Universidad Nacional Aut\'onoma de M\'exico,
\\[-1.1mm]
\small
C.P. 04510, M\'exico D.F.
\\[-1.1mm]
\small
and
\\[-1.1mm]
\small
 International Research Laboratory ``Multiscale Model Reduction",
\\[-1.1mm]
\small
Ammosov North-Eastern Federal University,
\\[-1.1mm]
\small
Yakutsk, Russia
\\[-1.1mm]
\small
\texttt{alexander.v.kiselev@gmail.com}
\\[2mm]
\textbf{Luis O. Silva}
%\thanks{Partially supported by UNAM-DGAPA-PAPIIT IN105414}
\\
%% ----- Institution --------
\small
Departamento de F\'{i}sica Matem\'{a}tica,\\[-1.2mm]
\small
Instituto de Investigaciones en Matem\'aticas Aplicadas y en Sistemas,\\[-1.1mm]
\small
Universidad Nacional Aut\'onoma de M\'exico,\\[-1.1mm]
\small
C.P. 04510, M\'exico D.F.\\[-1.1mm]
\small
\texttt{silva@iimas.unam.mx}
}
%%%%%%%%
\date{}
\maketitle
\vspace{-4mm}
\begin{center}
  \textsl{To the fond memory of Professor Boris Pavlov}
\end{center}
\begin{center}
\begin{minipage}{5in}
  \centerline{{\bf Abstract}} \bigskip This work deals with the
  functional model for a class of extensions of symmetric operators and its
  applications to the theory of wave scattering. In terms of Boris
  Pavlov's spectral form of this model, we find explicit formulae for
  the action of the unitary group of exponentials corresponding to
  {\it almost solvable} extensions of a given closed symmetric operator with
  equal deficiency indices. On the basis of these formulae, we are
  able to construct wave operators and derive a new representation for the scattering matrix for
  pairs of such extensions in both self-adjoint and non-self-adjoint situations.  %We use this representation to {\it
%    explicitly} recover the coupling constants in the inverse
%  scattering problem for a finite non-compact quantum graph with
%  $\delta$-type vertex conditions.
\end{minipage}
\end{center}
%\thispagestyle{empty}
%\end{titlepage}
%%%%%%%%%%%%%%%%%%%%%%%%%%%%%%
\section{Introduction}
\label{sec:introduction}

Over the last eighty years or so, the subject of the mathematical
analysis of waves interacting with obstacles and structures
(``scattering theory'') has served as one of the most impressive
examples of bridging abstract mathematics and applications to physics,
which in turn motivated the development of new mathematical
techniques. The pioneering works of von Neumann \cite{MR1503053},
\cite{MR0066944} and his contemporaries during 1930--1950, on the mathematical foundations of quantum mechanics, fuelled the interest of
mathematical analysts to formulating and addressing the problems of
direct and inverse wave scattering in a rigorous way.

The foundations of the modern mathematical scattering theory were laid
by Friedrichs, Kato and
%\cite{MR0407617} and developed by
Rosenblum \cite{MR0407617,
  MR0090028, Friedrichs} and subsequently by Birman and Kre\u\i n
\cite{MR0139007},  Birman \cite{Birman_1963}, Kato and Kuroda \cite{MR0385604} and Pearson
\cite{MR0328674}. For a detailed exposition of this subject, see
\cite{MR529429, MR1180965}.
%, Ostrovsky
%\cite{Ostrovsky};
%see also a comprehensive review by Faddeev \cite{MR0110466, MR0149843}.

The direct and inverse scattering  on the
infinite and semi-infinite line was extensively studied using the classical
integral-operator techniques by Borg \cite{MR0015185, MR0058063},
Levinson \cite{MR0032067}, Krein \cite{MR0039895, MR0058072,
  MR0078543}, Gel'fand and Levitan \cite{MR0045281}, Marchenko
\cite{MR0075402}, Faddeev \cite{MR0149843, Faddeev_additional},  Deift and Trubowitz
\cite{MR0622619}. In this body of work, the crucial role is played by
the classical Weyl-Titchmarsh $m$-coefficient.

In the general operator-theoretic context, the $m$-coefficient is
generalised to both the classical Dirichlet-to-Neumann map (in the PDE
setting; cf. also \cite{Amrein}), and to the so-called $M$-operator, which takes the form of
the
%generalised
Weyl-Titchmarsh $M$-matrix in the case of symmetric operators with equal deficiency
indices.  This has been exploited extensively in the study of
operators, self-adjoint and non-selfadjoint alike, through the works 
in Ukraine (brought about by the influence of M.Kre\u\i n) on the theory of boundary triples and the
associated $M$-operators (Gorbachuk and Gorbachuk \cite{MR1294813},
Ko\v cube\u\i\ \cite{MR0365218, MR0592863}, Derkach and Malamud
\cite{MR1087947} and further developments) and of the students of Pavlov in
St.\,Petersburg
%on the derivation and analysis of functional models
%for various classes of non-selfadjoint operators and of the associated
%formulae for wave operators
(see {\it e.g.}  \cite{MR2330831,Kiselev,MR2418300}).

A parallel approach, which provides a
connection to the theory of dissipative operators, was developed by
Lax and Phillips \cite{MR0217440}, who analysed the direct scattering
problem for a wide class of linear operators in the Hilbert space,
including those associated with the multi-dimensional acoustic problem
outside an obstacle, using the language of group theory (and, indeed,
thereby developing the semigroup methods in operator theory). The
associated techniques were also termed ``resonance scattering'' by Lax
and Phillips.

By virtue of the underlying dissipative framework, the above activity
set the stage for the applications of non-selfadjoint techniques, in
particular for the functional model for contractions and dissipative
operators by Sz\"{o}kefalvi-Nagy and Foia\c{s} \cite{MR2760647}, which
has shown the special r\^{o}le in it of the characteristic function of Liv\v{s}ic
\cite{MR0020719} and allowed Pavlov \cite{MR0510053} to construct a
spectral form of the functional model for dissipative operators.  The
connection between this work and the concepts of scattering theory was
uncovered by the famous theorem of Adamyan and Arov
\cite{MR0206711}. In a closely related development, Adamyan and Pavlov \cite{AdamyanPavlov}
established a description for the scattering matrix of a pair of self-adjoint extensions of a symmetric operator (densely or non-densely defined) with finite equal deficiency indices.

Further, Naboko \cite{MR573902} advanced the
research initiated by Pavlov, Adamyan and Arov in two
directions. Firstly, he generalised Pavlov's construction to the case
of non-dissipative operators, and secondly, he %bridged the gap back to
%the mathematical scattering theory. In particular, he
provided
explicit formulae for the wave operators and scattering matrices of a
pair of (in general, non-selfadjoint) operators in the functional
model setting. It is remarkable that in this work of Naboko the
difference between the so-called stationary and non-stationary
scattering approaches disappears.

There exists a wide body of work, carried out in the last thirty years or so, dedicated to the analysis of the scattering theory for general non-selfadjoint operators \cite{Mak_Vas, Solomyak, Tikhonov, Ryzh_ac_sing, Ryzhov_equipped, Ryzhov_closed}. These works make a substantial use of functional model techniques in the non-selfadjoint case and provide the most general results, without taking into account the specific features of any particular subclass of operators under consideration. In particular, the paper \cite{Ryzhov_equipped} essentially generalises to the non-selfadjoint case the classical stationary approach to the construction of wave operators \cite{MR1180965}. On the other hand, as pointed out above, the study of non-selfadjoint extensions of symmetric operators naturally lends itself to the use of the theory of boundary triples and associated $M$-operators, thus taking advantage of the concrete properties of this subclass. This has been exploited in \cite{MR2330831}, where a functional model for dissipative and non-dissipative almost solvable extensions of symmetric operators was developed in terms of the theory of boundary triples. This work, however, stops short of 
%constructing the functional model in the non-dissipative case (although it has to be noted that the corresponding construction can be easily %obtained, based on the general results contained in the above list of works) and of 
the characterisation of the absolutely continuous subspace of the operator considered in the ``natural'' terms associated with boundary triples and $M$-operators ({\it cf.} \cite{Ryzh_ac_sing, Romanov}, where the concept of the absolutely continuous subspace of a self-adjoint operator is discussed in the most general case). If one bridges this (in fact, very narrow) gap, as we do in Sections \ref{sec:resolvent-extensions}, \ref{sec:functional-model}, this opens up a possibility to directly apply Naboko's argument \cite{MR573902}, which then yields both the explicit expression for wave operators and concise, easily checked sufficient conditions for the existence and completeness of wave operators, formulated in natural terms. What is more, it also yields an explicit expression for the scattering matrix of the problem, formulated in terms of the $M$-operator and parameters fixing the extension.

Our aim in the present work is therefore twofold: first, it is to expose the methodology of functional model in application to the development of scattering theory for non-selfadjoint operators and, second, to apply this methodology to the case of almost solvable extensions of symmetric operators, yielding new, concise and explicit, results in the special and important in applications case. With this aim in mind, we endeavour to extend the approach of Naboko
\cite{MR573902}, which was formulated for additive perturbations of
self-adjoint operators, to the case of \emph{both self-adjoint and non-self-adjoint} extensions of symmetric
operators, under the only additional assumption that this extension is almost solvable, see Section \ref{sec:boundary-triples} below for precise definitions. Unfortunately, the named assumption is rather restrictive in nature, see Remark \ref{rem:almost-solvable} below. Still, already the framework of almost solvable extensions allows us to consider direct and inverse scattering problems on quantum graphs, see
\cite{CherednichenkoKiselevSilva} for an application of abstract results of this paper in the mentioned setting. We also point out that the case we consider proves to be sufficiently generic to allow for a treatment of the scattering problem for models of double porosity in homogenisation, see \cite{CEK, CEKN}.

% In pursuing the above aim, we will use a version of the functional
%model of Pavlov and Naboko as developed by Ryzhov \cite{MR2330831}.
%The work \cite{MR2330831} stopped short of
%proving the crucial, from the scattering point of view, theorem on
%``smooth'' vectors and therefore was unable to extend Naboko's
%results on the scattering theory to the setting of (in general,
%non-selfadjoint) extensions of symmetric operators.

The paper is organised as follows. In Section
\ref{sec:boundary-triples} we recall the key points of the theory of
boundary triples for extensions of symmetric operators with equal
deficiency indices and introduce the associated $M$-operators,
following mainly \cite{MR1087947} and \cite{MR2330831}. In Section
\ref{sec:resolvent-extensions} we derive formulae for the resolvents
of the family of extensions $A_\varkappa$ parametrised by operators
$\varkappa$ in the boundary space, in terms of the so-called
characteristic function of a fixed element of the family. These
formulae are then employed in Section \ref{sec:functional-model} to
derive the functional model for the above family of extensions. The
material of Sections \ref{sec:resolvent-extensions} and
\ref{sec:functional-model} closely follows the approach of
\cite{MR2330831} and is based on the much more general facts of {\it e.g.} \cite{MR573902, Mak_Vas, Ryzh_ac_sing, Ryzhov_closed}, and references therein. Moreover, although this functional model can be seen as a particular case of more general results of the above papers, it proves however much more convenient for our purposes, due to the fact that it is explicitly formulated in the natural, from the point of view of the operator considered, terms. In Section \ref{sec:inclusion-smooth} we
characterise the absolutely continuous subspace of $A_\varkappa$ as
the closure of the set of ``smooth'' vectors in the model Hilbert
space introduced in Section~\ref{sec:functional-model}. In doing so, we follow the general framework of \cite{Ryzh_ac_sing}, but, again, the fact that we use the specifics of a particular class of non-selfadjoint operators allows us to obtain this characterisation in a concise, easily usable form. On this basis, in Section~\ref{sec:wave-operators} we
define the wave operators for a pair from the family $\{A_\varkappa\}$
and demonstrate their completeness property under natural, easily verifiable assumptions. This, in combination with the functional model, allows us to obtain formulae for the scattering
operator of the pair. In Section~\ref{sec:spectral-repr-ac} we describe the representation of the
scattering operator as the scattering
matrix, which is explicitly written in terms of the $M$-operator.
%,
%analogous to the classical notion of the scattering matrix.
%All material up to this point is applicable to a general class of
%operators subject to the assumptions discussed in
%Sections~\ref{sec:boundary-triples} and \ref{sec:inclusion-smooth}. In
%Section~\ref{sec:quantum-graphs} we recall the concept of a quantum
%graph and discuss the implications of the preceding theory for the
%associated scattering operator for the pair $(A_\varkappa, A_0),$
%where $\varkappa$ is the parametrising operator as before, now written
%in terms of the ``coupling'' constants at the graph vertices and
%$A_0=A_\varkappa\vert_{\varkappa=0}$ is the ``unperturbed'' operator
%with Kirchhoff vertex conditions. Finally, in
%Section~\ref{sec:inverse-scattering} we solve the inverse scattering
%problem for a graph with $\delta$-type couplings at the vertices,
%using the formulae for the scattering matrix in terms of the
%$M$-matrix of the graph.

%\section{Preliminaries}
%\label{sec:preliminaries}

\section{Extension theory and boundary triples}
\label{sec:boundary-triples}
Let $\cH$ be a separable Hilbert space and denote by
$\inner{\cdot}{\cdot}$ the inner product in this space. % (which we
%consider to be antilinear in the second argument).

%An important subclass of the class of dissipative operators is the
%class of symmetric operators, for which the equality in
%\eqref{eq:definition-disipactive} holds.
%Thus, if $A$ is a symmetric
%operator, then $A\subset A^*$.
Let $A$ be a closed symmetric operator densely defined in $\cH,$ {\it
  i.e.} $A\subset A^*,$ with domain $\dom(A)\subset\cH.$
%For a symmetric operator $A$,
%%For such operators, all points in the lower and
%%upper half-planes are of regular type and
%For a closed symmetric operator $A$,
The deficiency indices $ n_+(A), n_-(A)$ are defined as follows:
\begin{equation*}
  n_\pm(A):=\dim(\cH\ominus\ran(A-z I))=\dim(\ker(A^*-\cc{z}I))\,,\quad z\in\complex_\pm\,.
\end{equation*}
%%If $A=A^*$ then $A$ is referred to as self-adjoint.
%%\begin{definition}
%%\label{def:completelynon-selfadjoint}
A closed
%dissipative
operator $L$ is said to be \emph{completely non-selfadjoint} if there
is no subspace reducing $L$ such that the part of $L$ in this subspace is self-adjoint. A completely non-selfadjoint symmetric operator
is often referred to as \emph{simple}.
%\end{definition}
%For a closed symmetric operator to be simple it
%suffices that it has no invariant subspace in which the
%operator is self-adjoint \cite[Thm.\, 4.6.1]{MR1192782}.

As shown in \cite[Sec.\,1.3]{MR0048704}(see also
\cite[Thm.\,1.2.1]{MR1466698}), the maximal invariant subspace for the
closed symmetric operator $A$ in which it is self-adjoint is
%\begin{equation*}
  $\bigcap_{z\in\complex\setminus\reals}\ran(A-zI)\,.$
%\end{equation*}
Thus, a necessary and sufficient condition for the closed symmetric
operator $A$ to be completely non-selfadjoint (or simple) is that
\begin{equation}
\label{eq:completely-non-selfadjoint-criterion}
  \bigcap_{z\in\complex\setminus\reals}\ran(A-zI)=\{0\}\,.
\end{equation}

In this work we consider extensions of a given closed symmetric
operator $A$ with equal deficiency indices, {\it i.\,e.} $n_-(A)=n_+(A)$,
and use the theory of boundary triples.
 %Section~\ref{sec:boundary
  %triples and extension theory})
In order to deal with the family of extensions $\{A_\varkappa\}$ of
the symmetric operator $A$ (where the parameter $\varkappa$ is itself an operator, see notation
immediately following Proposition \ref{prop:weyl-function-properties}), we first construct a functional model of its
particular dissipative extension. This is done following the
Pavlov-Naboko procedure, which in turn stems from the functional model of Sz\"{o}kefalvi-Nagy and Foia\c{s}. This allows us to obtain a simple model for the
whole family $\{A_\varkappa\},$ in particular yielding a possibility
to apply it to the scattering theory for certain pairs of operators in
$\{A_\varkappa\}$, for both cases when these operators are
self-adjoint and non-selfadjoint, including the possibility that both operators of the pair are non-selfadjoint.

%%The parameter $\varkappa$ is itself an operator, see notation
%%immediately following Proposition \ref{prop:weyl-function-properties}.
%Notation~\ref{not:kappa-alpha-b}
%and \ref{not:Akappa} below).
%%For the whole family of extensions $\{A_\varkappa\}$, we construct a
%%functional model based on Pavlov-Naboko construction which in turn
%%stems from Sz.-Nagy-Foia\c{s} functional model. Having a single model
%%for the family $\{A_\varkappa\}$, allows us to apply it to the
%%scattering theory for some pairs of operators in $\{A_\varkappa\}$,
%%including both the cases when these operators are self-adjoint and
%%non-selfadjoint.

Taking into account the importance of dissipative operators in our
work, we briefly recall that a densely defined operator $L$ in $\cH$
is called dissipative if
\begin{equation}
  \label{eq:definition-disipactive}
  \im\inner{Lf}{f}\ge 0\quad \quad\forall f\in\dom(L).
\end{equation}
% For a dissipative operator $L$, the lower half-plane is contained in the set of points of regular type, {\it i.e.}
%\begin{equation*}
%  \complex_-\subset\{z\in\complex: \exists C>0\ \ \forall f\in\dom(L)\ \ \norm{(L-zI)f}\ge
%  C\norm{f}\}\,.
%\end{equation*}
A dissipative operator $L$ is called maximal if $\complex_-$ is
 contained in its resolvent set
$\rho(L):=\{z\in\complex: (L-zI)^{-1}\in\mathcal{B}(\cH)\}$.
($\mathcal{B}(\cH)$ denotes the space of bounded operators defined on
the whole Hilbert space $\cH$). Clearly, a maximal dissipative
operator is closed; any dissipative operator admits a maximal extension.

We next describe the boundary triple approach to
the extension theory of symmetric operators with equal deficiency
indices (see in \cite{Derkach} a review of the subject).
This approach has proven to be  particularly useful in the study of
self-adjoint extensions of ordinary differential operators of second order.

%In the Hilbert space $\cH$, consider a closed symmetric operator $A$
%with equal deficiency indices.

\begin{definition}
  \label{def:boudary-triple}
  For a closed symmetric operator $A$ with equal deficiency indices, consider the linear mappings
  \begin{align*}
    \Gamma_1:\dom(A^*)\to\mathcal{K},\ \ \ \ \ \
\Gamma_0:\dom(A^*)\to\mathcal{K}\,,
  \end{align*}
where $\mathcal{K}$ is an auxiliary separable Hilbert space, such that
\begin{align}
(1)&\quad
  \inner{A^*f}{g}_\cH-\inner{f}{A^*g}_\cH=
\inner{\Gamma_1f}{\Gamma_0g}_\mathcal{K}-\inner{\Gamma_0f}{\Gamma_1g}_\mathcal{K};\label{Green_formula}\\[0.4em]
(2)&\quad \text{The mapping }\dom(A^*)\ni f\mapsto
\binom{\Gamma_1f}{\Gamma_0f}\in\mathcal{K}\oplus\mathcal{K}\text{ is surjective.}\nonumber
\end{align}
Then the triple $(\mathcal{K},\Gamma_1,\Gamma_0)$ is said to be a \emph{boundary
  triple} for $A^*$.
\end{definition}
\begin{remark}
  \label{rem:equal-indices-existence-triple}
  There exist boundary triples for $A^*$ whenever $A$ has equal
  deficiency indices (the case of infinite indices is not excluded), see \cite[Theorem 3]{MR0365218}.
\end{remark}
In this work we consider \emph{proper extensions} of
$A$, {\it i.e.} extensions of $A$ that are restrictions of $A^*$.
%For a closed linear relation $B$ on $\mathcal{K},$ {\it i.e.} a
%subspace of $\mathcal{K}\oplus\mathcal{K}$, let $A_B$ be the
%restriction of $A^*$ such that, for a specific choice of the triple
%$(\mathcal K, \Gamma_1, \Gamma_0)$ one has
%\begin{equation}
%\label{eq:extension-by-relation}
%  \dom(A_B)=\left\{f\in\dom(A^*):\binom{\Gamma_1f}{\Gamma_0f}\in B\right\}\,.
%\end{equation}
%Clearly, $A_B$ is a proper extension of $A$ (see \cite[Sec.~14]{MR2953553}).
%In this paper we treat those proper extensions $A_B$ of $A$ that arise from linear relations $B$ that are graphs of bounded operators defined on the whole space $\mathcal{K}.$ In this case we identify the relation $B$ and the corresponding bounded operator, {\it i.\,e.,} $B\in\mathcal{B}(\mathcal{K}).$
The extensions $A_B$ for which there exists a triple $(\mathcal K, \Gamma_1, \Gamma_0)$
and $B\in\mathcal{B}(\mathcal{K})$ such that
\begin{equation}
\label{eq:extension-by-operator}
  f\in\dom(A_B) \iff \Gamma_1f=B\Gamma_0f\,.
\end{equation} are called
\emph{almost solvable}
with respect to the triple $(\mathcal K,
\Gamma_1, \Gamma_0)$.
% In this case, (\ref{eq:extension-by-relation}) implies that

%Note that if an operator is almost solvable depends on the choice
%of boundary triple and is not a feature of the extension alone.

\begin{remark}
  \label{rem:almost-solvable}
  Admittedly, the framework of almost solvable extensions is quite restrictive. In particular, even the standard three-dimensional scattering problem for PDEs in an exterior domain, with classical boundary condition (self-adjoint and non-selfadjoint alike) cannot be treated using this approach, see the discussion in \cite{MR2418300} and also references therein. It would appear that one needs to employ the more general setting of linear relations \cite{MR1154792}, in order to accommodate this problem. However, the named setting is substantially more involved and complex than the theory of almost solvable extensions, so that the  blueprints of the Sz-Nagy--Foia\c{s} model of closed linear relations do not seem to be available as of today.

  On the other hand, there exist at least two recent developments suggesting that the approach of the present paper can be extended beyond the natural limitations of the theory of almost solvable extensions. These are, firstly, the work \cite{Ryzh_spec}, which offers a unified operator-theoretic approach to boundary-value problems and, in particular, an abstract definition of the $M$-operator suitable for the construction of a functional model; and secondly, the recent paper \cite{BMNW2018}, which provides an explicit form of a functional model for PDE problems associated with dissipative operators. We hope to pursue this rather intriguing subject elsewhere.

  %%Most of what is said henceforth
%  this paper, and also in \cite{MR2330831}, should be
  %%remains valid not only for almost solvable extensions but
  %not only for almost solvable extensions but
  %%for more general extensions given by so-called closed linear relations \cite{MR1154792}, which we shall pursue elsewhere.
\end{remark}

The following assertions, written in slightly different terms, can be
found in \cite[Thm.\,2]{MR0365218} and \cite[Chap.\,3
Sec.\,1.4]{MR1154792} (see also %\cite[Thm.\,2.3]{MR2732083},
\cite[Thm.\,1.1]{MR2330831}, and \cite[Sec.~14]{MR2953553} for an alternative formulation). We compile them in the next
proposition for easy reference.

\begin{proposition}
  \label{prop:properties-almost-extensions}
   Let $A$ be a closed symmetric operator with equal deficiency
  indices and let $(\mathcal{K},\Gamma_1,\Gamma_0)$ be a the boundary
  triple for $A^*$. Assume that $A_B$ is an almost solvable
  extension. Then the following statements hold:
  \begin{enumerate}
  \item $f\in\dom(A)$ if and only if $\Gamma_1f=\Gamma_0f=0.$
  \item $A_B$ is maximal, i.\,e., $\rho(A_B)\ne\emptyset$.
  \item $A_B^*=A_{B^*}.$
  \item $A_B$ is dissipative if and only if $B$ is dissipative.
  \item $A_B$ is self-adjoint if and only if $B$ is self-adjoint.
  \end{enumerate}
\end{proposition}

\begin{definition}
  \label{def:weyl-function}
  The function $M:\complex_-\cup\complex_+\to \mathcal{B}(\cH)$ such
  that
  \begin{equation*}
    M(z)\Gamma_0f=\Gamma_1f\ \ \ \ \ \forall f\in\ker(A^*-zI)
  \end{equation*}
%for any
is \emph{the Weyl function of the boundary triple}
  $(\mathcal{K},\Gamma_1,\Gamma_0)$ for $A^*,$ where $A$ is assumed to be as in Proposition \ref{prop:properties-almost-extensions}.
 \end{definition}

The Weyl function defined above has the following properties \cite{MR1087947}.
\begin{proposition}
  \label{prop:weyl-function-properties}
  Let $M$ be a Weyl function of the boundary triple
  $(\mathcal{K},\Gamma_1,\Gamma_0)$ for $A^*,$ where $A$ is a closed symmetric operator with equal deficiency indices. Then the following statements hold:
  \begin{enumerate}
  \item $M:\complex\setminus\reals\to\mathcal{B}(\mathcal{K})$\,.
  \item $M$ is a $\mathcal{B}(\mathcal{K})$-valued double-sided
    $\mathcal{R}$-function \cite{MR0328627KK}, that is,
    \begin{equation*}
      M(z)^*=M(\cc{z})\quad\text{ and }\quad\im(z)\im(M(z))>0\quad\text{ for }
      z\in\complex\setminus\reals\,.
    \end{equation*}
  \item The spectrum of $A_B$ coincides with the set of points $z_0\in{\mathbb C}$ such
    that $(M-B)^{-1}$
    %as a function of $z,$
    does not admit analytic continuation into $z_0.$
    %have a bounded inverse on $\mathcal K.$
  \end{enumerate}
\end{proposition}

%\begin{notation}
%\label{not:kappa-alpha-b}
Let us lay out the notation for some of the main objects in this
paper.  In the auxiliary Hilbert space $\mathcal{K}$, choose
  %a self-adjoint operator $B$ and
  a bounded nonnegative self-adjoint operator $\alpha$ so
  %in such a way
  that the operator
 \begin{equation}
    \label{eq:b-kappa-def}
    B_\varkappa:=\frac{\alpha\varkappa\alpha}{2}
  \end{equation}
 belongs to $\mathcal{B}(\mathcal{K})$, where $\varkappa$ is a bounded operator in
  $E:=\clos(\ran(\alpha))\subset\mathcal{K}.$
%Here we also assume that
%$\alpha\varkappa\alpha/2$ is a relatively bounded perturbation of $B$. Clearly, if $B$ is in
%$\mathcal{B}(\mathcal{K})$, then $\alpha\varkappa\alpha/2$ belongs to
%$\mathcal{B}(\mathcal{K})$.
%\end{notation}
%\textcolor{red}{This setting is the most general (one may have
%  unbounded operators $B$ and $\alpha\varkappa\alpha/2$ whose sum is
%  bounded). However, it is advantageous (see below) to assume one of
%  these operators (and the the other one) is bounded}
In what follows, we deal with almost solvable extensions of a given
symmetric operator $A$ that are generated by $B_\varkappa$ via (\ref{eq:extension-by-operator}).
We always assume that the deficiency indices of $A$ are equal and that some boundary triple
%an arbitrary, but fixed, boundary triple
$(\mathcal{K},\Gamma_1,\Gamma_0)$ for
$A^*$ is fixed.
%\begin{notation}
%\label{not:Akappa}
In order to streamline the
 formulae, we write
 \begin{equation}
   \label{eq:a-kappa-def}
   A_\varkappa:=A_{B_\varkappa}.
 \end{equation}
%\end{notation}
Here $\varkappa$ should be understood as a parameter for a family of
almost solvable extensions of $A$. Note that if $\varkappa$ is
self-adjoint then so is $B_\varkappa$ and, hence by
Proposition~\ref{prop:properties-almost-extensions}(5), $A_\varkappa$ is
self-adjoint. Note also that $A_{\I I}$ is maximal dissipative, again by
Proposition~\ref{prop:properties-almost-extensions}.

\begin{definition}
  \label{def:charasteristic-function}
  The characteristic function of the operator $A_{\I I}$ is
  the operator-valued function $S$ on $\complex_+$ given by
 %such that
  \begin{equation}
    S(z):=I\eval{E}+\I\alpha\bigl(B_{\I I}^*-M(z)\bigr)^{-1}\alpha\eval{E},\ \ \ \ \ z\in\complex_+.
    \label{S_definition}
  \end{equation}
 \end{definition}
In the general setting, the characteristic function is defined as in
\cite[Def.\,1.7]{MR2330831}. Our definition is justified by
\cite[Eq.\,1.16]{MR2330831}.
%Note that, strictly speaking, in the case of unbounded $\alpha$ the operator $S(z)$ is defined on the set $E\cap\dom(\alpha)$ and then extended by continuity to $E.$
 \begin{remark}
  \label{rem:def-charasteristic-function}
  The function $S$ is analytic in $\complex_+$ and, for each
  $z\in\complex_+$, the mapping $S(z):E\to E$ is a contraction. Therefore, $S$ has
  nontangential limits almost everywhere on the real line in the
  strong topology \cite{MR2760647},
%  [Sec.\,1]{MR573902},
  which we will henceforth denote by $S(k),$ $k\in{\mathbb R}.$
\end{remark}
\begin{remark}
When
  $\alpha=\sqrt{2}I$, an straightforward calculation yields that
  $S(z)$ is the Cayley transform of $M(z)$, {\it i.e.}
 \begin{equation*}
   S(z)=(M(z)-\I I)(M(z)+\I I)^{-1}\,.
 \end{equation*}
\end{remark}

\section{Formulae for the resolvents of  almost solvable extensions}
\label{sec:resolvent-extensions}
In this section we establish some useful relations between the
resolvents of the operators $A_\varkappa$ for any
$\varkappa\in\mathcal{B}(E)$ and the resolvents of the maximal
dissipative operator $A_{\I I}$ and its adjoint. These relations ({\it cf.} \cite{Ryzhov_equipped, Ryzh_ac_sing, Ryzhov_closed} and references therein, for the corresponding results in the general setting of closed non-selfadjoint operators) are instrumental for the construction of the functional model in the next section.
\begin{notation}
  \label{not:theta-functions}
We abbreviate
\begin{align}
  \Theta_\varkappa(z):&=I-\I\alpha(B_{\I I}-M(z))^{-1}
\alpha\chi_\varkappa^+\,,\qquad z\in\complex_-\,,\label{eq:theta}\\
  \widehat{\Theta}_\varkappa(z):&=I+\I\alpha(B_{\I I}^*-M(z))^{-1}
\alpha\chi_\varkappa^-\,,\qquad z\in\complex_+\,,\label{eq:theta-hat}
\end{align}
where
\begin{equation}
  \label{eq:chi-def}
  \chi_\varkappa^\pm:=\frac{I\pm\I\varkappa}{2},
\end{equation}
and for simplicity we have written $I$ instead of $I\eval{E}$. We use
this convention throughout the text.
\end{notation}

It follows from Definition~\ref{def:charasteristic-function} and
Proposition~\ref{prop:weyl-function-properties}(2) that
the operator-valued functions $\Theta_\varkappa(z)$ and
  $\widehat{\Theta}_\varkappa(z)$ can be expressed in terms of the characteristic function $S,$ as follows:
  \begin{align}
   \Theta_\varkappa(z)&=I+(S^*(\cc{z})-I)\chi_\varkappa^+\quad\quad
   %\text{ for all }
   \forall\,z\in\complex_-\,,\label{eq:alternate-for-theta}\\[0.4em]
   \widehat{\Theta}_\varkappa(z)&=I+(S(z)-I)\chi_\varkappa^-\,\quad\quad
   \,
   %\text{ for all }
   \forall\,z\in\complex_+\,.\label{eq:alternate-for-theta-hat}
  \end{align}
The formulae in the next lemma are analogous to
\cite[Eqs. 2.18 and 2.22]{MR2330831}.
\begin{lemma}
  \label{lem:ryzhov-formulae}
  The following identities hold:
  \begin{enumerate}[(i)]
\item
    $\alpha\Gamma_0(A_{\I I}-zI)^{-1}
=\Theta_\varkappa(z)\alpha\Gamma_0(A_{\varkappa}-zI)^{-1}\quad\ \ \forall\,z\in\complex_-\cap\rho(A_\varkappa)$;\label{eq:ii-kappa}
\item
 $\alpha\Gamma_0(A_{\varkappa}-zI)^{-1}
=\Theta_\varkappa(z)^{-1}\alpha\Gamma_0(A_{\I I}-zI)^{-1}\quad\ \ \forall\,z\in\complex_-\cap\rho(A_\varkappa)$;\label{eq:kappa-ii}
\item
 $\alpha\Gamma_0(A_{\I I}^*-zI)^{-1}
=\widehat{\Theta}_\varkappa(z)\alpha\Gamma_0(A_{\varkappa}-zI)^{-1}\quad\ \ \forall\,z\in\complex_+\cap\rho(A_\varkappa)$;\label{eq:ii-kappa+}
\item
$\alpha\Gamma_0(A_{\varkappa}-zI)^{-1}
=\widehat{\Theta}_\varkappa(z)^{-1}\alpha\Gamma_0(A_{\I
  I}^*-zI)^{-1}\quad\ \ \forall\,z\in\complex_+\cap\rho(A_\varkappa)$\label{eq:kappa-ii+}\,.
\end{enumerate}
\end{lemma}
\begin{proof}
  We start by proving (\ref{eq:ii-kappa}). To this end,
  suppose that $z\in\complex_-\cap\rho(A_\varkappa)$ so $(A_{\I I}-z
  I)^{-1}$ and $(A_\varkappa-z I)^{-1}$ are defined
  on the whole space $\cH$. Fix an arbitrary $h\in\cH$ and define
\begin{equation}
    \label{eq:phi-g-def}
% \begin{split}
    \varphi:=(A_{\I I}-zI)^{-1}h,\ \ \ \ \ \ \ \ g:=(A_\varkappa-zI)^{-1}h\,.
 % \end{split}
\end{equation}
Clearly, the vector
  \begin{equation*}
    f:=\varphi-g=\left((A_{\I I}-zI)^{-1}-(A_{\varkappa}-zI)^{-1}\right)h
  \end{equation*}
is in $\ker(A^*-zI)$ since $A^*$ is an extension of both operators
$A_{\I I}$ and $A_{\varkappa}$. According to
\eqref{eq:extension-by-operator}, it follows from $\varphi\in\dom(A_{\I I})$ and $g\in\dom(A_\varkappa)$  that
$
%\begin{equation*}
  \Gamma_1\varphi=B_{\I I}\Gamma_0\varphi$ and $\Gamma_1g=B_\varkappa\Gamma_0g.$
%\end{equation*}
Thus, one has
\begin{align*}
  0&=\Gamma_1(f+g)-B_{\I I}\Gamma_0(f+g)\\
  &= \Gamma_1f-B_{\I I}\Gamma_0f+\Gamma_1g-B_{\I I}\Gamma_0g\\
  &=M(z)\Gamma_0f-B_{\I I}\Gamma_0f+ B_{\varkappa}\Gamma_0g-B_{\I I}\Gamma_0g\,,
\end{align*}
where in the last equality we also use the fact that $f\in\ker(A^*-zI),$
together with
Definition~\ref{def:weyl-function}.
% and that
%$g\in\dom(A_{\varkappa})$, together with \eqref{eq:extension-by-operator}.
Hence one has
\begin{equation*}
  \Gamma_0f=(B_{\I I}-M(z))^{-1}(B_{\varkappa}-B_{\I I})\Gamma_0g\,,
\end{equation*}
which, in turn, implies that
%, for $\Gamma_0\varphi=\Gamma_0f+\Gamma_0g$,
%the following holds
\begin{equation}
  \Gamma_0\varphi=\Gamma_0f+\Gamma_0g=\left[I+(B_{\I I}-M(z))^{-1}(B_{\varkappa}-B_{\I I})\right]\Gamma_0g.
  \label{both_sides}
\end{equation}
Taking into account (\ref{eq:phi-g-def}), using the fact that
%\begin{equation*}
  $B_{\varkappa}-B_{\I I}=-\I\alpha\chi_\varkappa^+\alpha$
%\end{equation*}
and applying the operator $\alpha$ to both sides of (\ref{both_sides}), we obtain
\begin{equation*}
  \alpha\Gamma_0(A_{\I I}-zI)^{-1}h=
\left[I-\I\alpha(B_{\I
    I}-M(z))^{-1}\alpha\chi_\varkappa^+\right]\alpha\Gamma_0
(A_{\varkappa}-zI)^{-1}h,
\end{equation*}
%after having applied $\alpha$ to both sides of the equation.
which is the identity (i), in view of the definition (\ref{eq:alternate-for-theta}).
%By substituting (\ref{eq:theta}) into this equation, one arrives at
%(\ref{eq:ii-kappa}).

Similar computations with the pairs $A_\varkappa, B_\varkappa$
and $A_{\I I}, B_{\I I}$ interchanged lead to
\begin{equation}
  \label{eq:almost-kappaii}
  \alpha\Gamma_0(A_{\varkappa}-zI)^{-1}h=
\left[I+\I\alpha(B_{\varkappa}-M(z))^{-1}\alpha\chi_\varkappa^+\right]
\alpha\Gamma_0(A_{\I I}-zI)^{-1}h,
\end{equation}
for $z\in{\mathbb C}_-\cap\rho(A_\varkappa).$ Now, (\ref{eq:kappa-ii}) follows
from (\ref{eq:almost-kappaii}) using the identity
\begin{equation}
 \label{eq:theta-inverse}
  \Theta_\varkappa(z)^{-1}=I+\I\alpha\bigl(B_\varkappa-M(z)\bigr)^{-1}\alpha\chi_\varkappa^+
  \quad\quad\forall\,z\in\complex_-\cap\rho(A_\varkappa)\,,
\end{equation}
which is validated by multiplying together the right-hand sides
of (\ref{eq:theta-inverse}) and \eqref{eq:theta} and employing a version
of the second resolvent identity
({\it cf.} \cite[Thm.\,5.13]{MR566954}):
\begin{equation*}
  (B_\varkappa-M(z))^{-1}-(B_{\I
    I}-M(z))^{-1}=(B_\varkappa-M(z))^{-1}(B_{\I I}-B_\varkappa)(B_{\I I}-M(z))^{-1}
\end{equation*}
which holds for all $z\in\complex_-\cap\rho(A_\varkappa)$.

We next proceed to the proof of \eqref{eq:ii-kappa+} and
\eqref{eq:kappa-ii+}. Fix an arbitrary
$z\in\complex_+\cap\rho(A_\varkappa)$ and an arbitrary $h\in\cH$ and define
\begin{equation}
\label{eq:varphi-g-defs+}
  \varphi:=(A_{\I I}^*-zI)^{-1}h\,,\qquad g:=(A_\varkappa-zI)^{-1}h\,,
\end{equation}
then $f:=\varphi-g$ is in $\ker(A^*-zI)$. Since
$\varphi\in\dom(A_{\I I}^*)$, one has that
\begin{align*}
  0&=\Gamma_1(f+g)-B_{\I I}^*\Gamma_0(f+g)\\[0.4em]
&=M(z)\Gamma_0f+ \Gamma_1g-B_{\I I}^*\Gamma_0f-B_{\I I}\Gamma_0g\,,
\end{align*}
where in the second equality
we use the fact
%have once again use
%Definition~\ref{def:weyl-function}
%taking into account
that
$f\in\ker(A^*-zI)$. On the other hand, in view of the inclusion
$g\in\dom(A_\varkappa)$, the formula \eqref{eq:extension-by-operator} allows us
to replace the second term in the last expression
% right-hand side of the last
%equation
by $B_\varkappa\Gamma_0g,$ which yields
\begin{equation}
 \label{eq:aux-iikappa-adjoint}
  0=(M(z)-B_{\I I}^*)\Gamma_0f+(B_\varkappa-B_{\I I}^*)\Gamma_0g\,.
\end{equation}
Since
%\begin{equation*}
  $B_\varkappa-B_{\I I}^*=\I\alpha\chi_\varkappa^-\alpha,$
%\end{equation*}
the equality~(\ref{eq:aux-iikappa-adjoint}) is rewritten as
\begin{equation*}
  \Gamma_0f=\I(B_{\I I}^*-M(z))^{-1}\alpha\chi_\varkappa^-\alpha\Gamma_0g\,,
\end{equation*}
which in turn implies that
\begin{equation*}
  \Gamma_0\varphi=\left[I+\I(B_{\I I}^*-M(z))^{-1}
\alpha\chi_\varkappa^-\alpha\right]\Gamma_0g\,.
\end{equation*}
Applying the operator $\alpha$ to both sides of the last equation and
using (\ref{eq:varphi-g-defs+}), we obtain
\begin{equation*}
  \alpha\Gamma_0(A_{\I I}^*-zI)^{-1}h=
  \left[I+\I\alpha(B_{\I I}^*-M(z))^{-1}\alpha\chi_\varkappa^-\right]
  \alpha\Gamma_0(A_\varkappa-zI)^{-1}h,
\end{equation*}
which is \eqref{eq:ii-kappa+}, in view of the definition (\ref{eq:alternate-for-theta-hat}).
% after substituting
%(\ref{eq:theta-hat}) into it.

Finally, we interchange the operators $A_{\I I}^*$ and $A_\varkappa$ in
(\ref{eq:varphi-g-defs+}) and repeat the computations, correspondingly
interchanging $B_{\I I}$ and $B_\varkappa.$ This yields the identity
\begin{equation}
 \label{eq:last-identity-theta-hat}
  \alpha\Gamma_0(A_{\varkappa}^*-zI)^{-1}h=
  \left[I-\I\alpha(B_{\I I}^*-M(z))^{-1}\alpha\chi_\varkappa^-\right]
  \alpha\Gamma_0(A_{\I I}^*-zI)^{-1}h,
\end{equation}
for all $z\in\complex_+\cap\rho(A_\varkappa).$ In a similar way to (\ref{eq:theta-inverse}), we verify that
\begin{equation*}
 \widehat{\Theta}_\varkappa(z)^{-1}=I-\I\alpha(B_{\I I}^*-M(z))^{-1}\alpha\chi_\varkappa^-\quad\quad \forall\,z\in\complex_+\cap\rho(A_\varkappa)
\end{equation*}
and hence
%, therefore, (\ref{eq:last-identity-theta-hat})
establish (\ref{eq:kappa-ii+}).
\end{proof}

\section{Functional model and theorems about smooth vectors}
\label{sec:functional-model}

Following \cite{MR573902}, we introduce a Hilbert space
serving as a functional model for the family of operators $A_\varkappa$. This
functional model was constructed for completely non-selfadjoint maximal
dissipative operators in \cite{MR0510053, MR0365199, Drogobych} and further
developed in \cite{MR573902, Ryzhov_closed, Ryzh_ac_sing, Tikhonov}. Next we recall some related necessary
information. In what follows, in various formulae, we use the
subscript ``$\pm$'' to indicate two different versions of the same formula in which the subscripts ``$+$'' and ``$-$'' are taken individually.

A function $f$ analytic on $\complex_\pm$ and taking values in $E$ is said to be
in the Hardy class $H^2_\pm(E)$ when
\begin{equation*}
  \sup_{y>0}\int_\reals\norm{f(x\pm\I y)}_E^2dx<+\infty
\end{equation*}
({\it cf.} \cite[Sec.\,4.8]{MR822228}). Whenever $f\in H^2_\pm(E)$, the
left-hand side of the above inequality defines $\norm{f}_{H^2_\pm(E)}^2$. We use the notation $H^2_+$ and $H^2_-$ for the usual Hardy spaces of ${\mathbb C}$-valued functions.

The elements of the Hardy spaces $H^2_\pm(E)$ are identified with
their boundary values,
 %in the topology of $E,$
 which exist almost
everywhere on the real line. We keep the same notation $H^2_\pm(E)$ for the corresponding  subspaces
of $L^2(\reals,E)$ \cite[Sec.\,4.8, Thm.\,B]{MR822228}).  By the
Paley-Wiener theorem \cite[Sec.\,4.8, Thm.\,E]{MR822228}), one
verifies that these subspaces are the orthogonal complements of each
other.
%({\it i.e.}, $L^2(\reals,E)={H}^2_+(E)\oplus {H}^2_-(E)$).

Following the argument of \cite[Thm.\,1]{MR573902}, it is shown in
\cite[Lem.\,2.4]{MR2330831} that
\begin{equation}
  \label{eq:naboko-thm-1}
  \alpha\Gamma_0(A_{\I I}-\cdot I)^{-1}h\in H_-^2(E)\quad\text{and}
  \quad\alpha\Gamma_0(A_{\I I}^*-\cdot I)^{-1}h\in H_+^2(E)\,.
\end{equation}

As mentioned in Remark~\ref{rem:def-charasteristic-function}, the
characteristic function $S$ given in
Definition~\ref{def:charasteristic-function} has nontangential limits
almost everywhere on the real line in the strong topology. Thus, for a
two-component vector function $\binom{\widetilde{g}}{g}$ taking values
in $E\oplus E$, one can consider the integral
\begin{equation}
  \label{eq:inner-in-functional}
  \int_\reals\inner{\begin{pmatrix} I & S^*(s)\\
    S(s) &
    I\end{pmatrix}\binom{\widetilde{g}(s)}{g(s)}}{\binom{\widetilde{g}(s)}{g(s)}}_{E\oplus
E}ds,
\end{equation}
which is always nonnegative, due to the contractive properties of $S.$
%(see Remark~\ref{rem:def-charasteristic-function}).
The space
\begin{equation}
\label{mathfrakH}
\mathfrak{H}:=L^2\Biggl(E\oplus E; \begin{pmatrix}
    I & S^*\\
    S & I
  \end{pmatrix}\Biggr)
\end{equation}
is the completion of the linear set of two-component vector functions
$\binom{\widetilde{g}}{g}: {\mathbb R}\to E\oplus E$ in the norm
(\ref{eq:inner-in-functional}), factored with respect to vectors of zero norm.
%We shall denote this space
%by $\mathfrak{H}$.
Naturally, not every element of the set can be identified with a pair
$\binom{\widetilde{g}}{g}$ of two independent functions. Still, in
what follows we keep the notation $\binom{\widetilde{g}}{g}$ for the
elements of this space.

  Another consequence of the contractive properties of the
  characteristic function $S$ is that for $\widetilde{g},g\in
  L^2(\reals,E)$ one has
  \begin{equation*}
    \norm{\binom{\widetilde{g}}{g}}_\mathfrak{H}\ge
    \begin{cases}
      \norm{\widetilde{g}+S^*g}_{L^2(\reals,E)},\\
      \norm{S\widetilde{g}+g}_{L^2(\reals,E)}\,.
    \end{cases}
  \end{equation*}
  Thus, for every Cauchy sequence
  $\{\binom{\widetilde{g}_n}{g_n}\}_{n=1}^\infty$, with respect to the $\mathfrak{H}$-topology,
  such that $\widetilde{g}_n,g_n\in L^2(\reals,E)$ for all
  $n\in\nats$, the limits of $\widetilde{g}_n+S^*g_n$ and
  $S\widetilde{g}_n+g_n$ exist in $L^2(\reals,E)$, so that $\widetilde{g}+S^*g$ and
  $S\widetilde{g}+g$ can always be treated as $L^2(\reals,E)$ functions.

Consider the orthogonal subspaces of $\mathfrak{H}$
\begin{equation}
  \label{eq:D-spaces}
  D_-:=
  \begin{pmatrix}
    0\\
    {H}^2_-(E)
  \end{pmatrix}\,,\quad
 D_+:=
  \begin{pmatrix}
   {H}^2_+(E)\\
   0
  \end{pmatrix}\,.
\end{equation}
We define the space
\begin{equation*}
  K:=\mathfrak{H}\ominus(D_-\oplus D_+),
\end{equation*}
which is characterised as follows (see {\it e.g.} \cite{MR0365199, Drogobych}):
\begin{equation}
  K=\left\{\begin{pmatrix}
   \widetilde{g}\\
   g
  \end{pmatrix}\in\mathfrak{H}: \widetilde{g}+S^*g\in  {H}^2_-(E)\,,
  S\widetilde{g}+g\in
 {H}^2_+(E)\right\}\,.
 \label{characterise_K}
\end{equation}
The orthogonal projection $P_K$ onto the subspace
$K$ is given by (see {\it e.g.} \cite{MR0500225})
\begin{equation}
\label{eq:pk-action}
  P_K
  \begin{pmatrix}
    \widetilde{g}\\
    g
  \end{pmatrix}
=
\begin{pmatrix}
  \widetilde{g}-P_+(\widetilde{g}+S^*g)\\[0.3em]
  g-P_-(S\,\widetilde{g}+g)
\end{pmatrix}\,,
\end{equation}
where $P_\pm$ are the orthogonal Riesz projections in $L^2(E)$ onto
${H}^2_\pm(E)$.

A completely
non-selfadjoint dissipative operator admits \cite{MR2760647} a
self-adjoint dilation.
%According to \cite{MR2330831},
The dilation $\mathscr{A}=\mathscr{A}^*$ of the operator $A_{\I I}$ is constructed following Pavlov's
procedure \cite{MR0365199, MR0510053, Drogobych}:  it is defined in the Hilbert space
\begin{equation}
  \label{eq:wider-hilbert-pavlov}
 % $
  \mathscr{H}=
  L^2(\reals_-,{\mathcal K})\oplus\mathcal{H}\oplus L^2(\reals_+,{\mathcal K}),
%  $
\end{equation}
%where $\cH$ is the original Hilbert space, see
%Section~\ref{sec:boundary-triples},
so that
\begin{equation*}
 P_\mathcal{H}(\mathscr{A}-zI)^{-1}\eval{\mathcal{H}}=
 (A_{\I I}-z I)^{-1}\,, \qquad z\in\complex_-.
\end{equation*}
%%and therefore
%%\begin{equation*}
%% P_\mathcal{H}(\mathscr{A}-zI)^{-1}\eval{\mathcal{H}}=
%% (A_{\I I}^*-z I)^{-1}\,,\qquad z\in\complex_+\,.
%%\end{equation*}
As in the case of additive non-selfadjoint perturbations
\cite{MR573902}, Ryzhov established in
\cite[Thm.\,2.3]{MR2330831} that $\mathfrak{H}$ serves as the
functional model for the dilation $\mathscr{A}$ {\it i.e.} there exists
an isometry $\Phi: \mathscr{H}\to\mathfrak{H},$ which we will make explicit below in our particular setting, such that
$\mathscr{A}$ is transformed into the operator of multiplication by
the independent variable:
%in $\mathfrak{H}:$
%\begin{equation}
 % \label{eq:multiplication-operator}
 $ \Phi(\mathscr{A}-z I)^{-1}=(\cdot-z)^{-1}\Phi\,.$
%\end{equation}
Furthermore, under this isometry %the space $\cH$ is mapped onto $K$:
%\begin{equation*}
  $$
  \Phi\eval\cH\cH=K
  $$
   unitarily, where $\cH$ is understood as being embedded in $\mathscr{H}$ in the natural way, {\it i.e.}
   $$
   \cH\ni h\mapsto 0\oplus h\oplus0\in\mathscr{H}.
   $$
  %$0\oplus\cH\oplus0\subset\mathscr{H}.$
  In what follows we keep the label $\Phi$ for the restriction $\Phi\eval\cH,$ in hope that it does not lead to confusion.

%that is, the space in which
%the simple closed symmetric operator $A$ with equal deficiency indices
%is densely defined.

The next theorem generalises \cite[Thm. 2.5]{MR2330831}, and its form is
similar to \cite[Thm.\,3]{MR573902}, which treats the case of additive
perturbations, see also \cite{Mak_Vas, MR2330831, Ryzh_ac_sing, Ryzhov_closed} for the case of possibly non-additive perturbations. The proof blends together the arguments of
\cite{MR2330831} and \cite{MR573902}, taking advantage of the
similarity between the formulae
(\ref{eq:theta})--(\ref{eq:alternate-for-theta-hat}) and those of
\cite[Section\,2]{MR573902}. It is standard, see {\it e.g.} \cite{Mak_Vas, MR573902, Ryzhov_closed}, and is therefore included in the Appendix for the sake of completeness only.
%%%The proof is given in the Appendix.
% for the sake of completeness.
\begin{theorem}
  \label{thm:rhyzhov2.5}
\begin{enumerate}[(i)]
\item  If $z\in\complex_-\cap\rho(A_\varkappa)$ and
  $\binom{\widetilde{g}}{g}\in K$, then
  \begin{equation}
    \Phi(A_\varkappa-z I)^{-1}\Phi^{*}\binom{\widetilde{g}}{g}=P_K\frac{1}{\cdot-z}
\binom{\widetilde{g}}{g-\chi_\varkappa^+\Theta^{-1}_\varkappa(z)(\widetilde{g}+S^*g)(z)}\,.
\label{representation_minus}
  \end{equation}
\item  If $z\in\complex_+\cap\rho(A_\varkappa)$ and
  $\binom{\widetilde{g}}{g}\in K$, then
  \begin{equation}
    \Phi(A_\varkappa-z I)^{-1}\Phi^{*}\binom{\widetilde{g}}{g}=P_K\frac{1}{\cdot-z}
\binom{\widetilde{g}-\chi_\varkappa^-\widehat{\Theta}^{-1}_\varkappa(z)(S\widetilde{g}+g)(z)}{g}\,.
 \label{representation_plus}
  \end{equation}
  Here, $(\widetilde{g}+S^*g)(z)$ and $(S\widetilde{g}+g)(z)$ denote
  the values at $z$ of the analytic continuations of the functions
  $\widetilde{g}+S^*g\in {H}^2_-(E)$ and $S\widetilde{g}+g\in
  {H}^2_+(E)$ into the lower half-plane and upper half-plane, respectively.
\end{enumerate}
\end{theorem}

% \begin{definition}
% \label{def:complete-non-selfadjointness}
% A closed symmetric non-selfadjoint operator is said to be completely
% non-selfadjoint if it is not a nontrivial orthogonal sum of a symmetric
% and a self-adjoint operators.
% \end{definition}

% In \cite{MR0048704} (see also \cite[Theorem 1.2.1]{MR1466698}), it is shown that
% \begin{equation*}
%  % \label{eq:maximal-invariant}
%   \bigcap_{z\in\complex\setminus\reals}\ran(A-zI)
% \end{equation*}
% is the maximal invariant subspace in which $A$ is self-adjoint. Hence,
% a necessary and sufficient condition for the symmetric operator $A$ to
% be completely non-selfadjoint is
% \begin{equation}
%   \label{eq:simplicity}
%   \bigcap_{z\in\complex\setminus\reals}\ran(A-zI)=\{0\}\,.
% \end{equation}
Following the ideas of Naboko, in the functional model space $\mathfrak{H}$
consider two subspaces $\mathfrak{N}^\varkappa_\pm$ defined as follows:
\begin{equation}
   \mathfrak{N}^\varkappa_\pm:=\left\{\binom{\widetilde{g}}{g}\in\mathfrak{H}:
  P_\pm\left(\chi_\varkappa^+(\widetilde{g}+S^*g)+\chi_\varkappa^-(S\widetilde{g}
  +g)\right)=0\right\}\,.
  \label{definition_curlyN}
\end{equation}
These subspaces have a characterisation in terms of the resolvent of the operator $A_\varkappa.$ This, again, can be seen as a consequence of a much more general argument (see {\it e.g.} \cite{Ryzhov_closed, Ryzh_ac_sing}). The proof in our particular case is provided in Appendix and follows the approach of \cite[Thm.\,4]{MR573902}.
%For completeness we provide its proof in the Appendix.
\begin{theorem}
  \label{lem:similar-to-naboko-thm-4}
Suppose that $\ker\{\alpha\}=0.$ The following characterisation holds:
\begin{equation}
 \mathfrak{N}^\varkappa_\pm=\left\{\binom{\widetilde{g}}{g}\in\mathfrak{H}:
  \Phi(A_{\varkappa}-z I)^{-1}\Phi^*P_K\binom{\widetilde{g}}{g}=
P_K\frac{1}{\cdot-z}\binom{\widetilde{g}}{g}
\text{ for all } z\in\complex_\pm\right\}\,.
\label{curlyN_plusminus}
\end{equation}
\end{theorem}
Consider the counterparts of $\mathfrak{N}^\varkappa_\pm$ in the original
Hilbert space $\cH:$
\begin{equation}
\label{eq:definition-n_pm}
  \widetilde{N}_\pm^\varkappa:=\Phi^*P_K\mathfrak{N}^\varkappa_\pm\,,
\end{equation}
which are linear sets albeit not necessarily subspaces.
%Following Naboko,
In a way similar to \cite{MR573902}, we introduce the set
\begin{equation*}
  \widetilde{N}_{\rm e}^\varkappa:=\widetilde{N}_+^\varkappa\cap
 \widetilde{N}_-^\varkappa
\end{equation*}
of so-called smooth vectors and its closure $N_{\rm e}^\varkappa:=\clos(\widetilde{N}_{\rm e}^\varkappa).$
%%These prove to be suitable for the model description of
%%the absolutely continuous subspace and, therefore, for the construction
%%of the wave operators.
%First,we introduce the set of smooth vectors
 In Section~\ref{sec:inclusion-smooth} we prove that $N_{\rm e}^\varkappa$
coincides with the absolutely continuous subspace of the operator
$A_\varkappa$ in the case when $A_\varkappa=A_\varkappa^*$ and under the same additional assumption that
${\rm ker}(\alpha)=\{0\},$ as in Theorem \ref{lem:similar-to-naboko-thm-4}.

%The reason for the choice of the term ``smooth vector"
%stems from the statement of Theorem~\ref{lem:similar-to-naboko-thm-4}
%together with (\ref{eq:definition-n_pm}).

% Note that if $\binom{\widetilde{g}}{g}\in
%\mathfrak{N}^\varkappa_+\cap\mathcal{N}_-^\varkappa$, then
% $\Phi^*P_K\binom{\widetilde{g}}{g}\in\widetilde{N}_{\rm e}^\varkappa$ and
% \begin{equation*}
% %  \label{eq:resolvent-in-N}
%   \Phi(A_{\varkappa}-z I)^{-1}\Phi^*P_K\binom{\widetilde{g}}{g}=
% P_K\frac{1}{\cdot-z}\binom{\widetilde{g}}{g}\quad\forall z\in\complex\setminus\reals\,,
% \end{equation*}

The next assertion ({\it cf. e.g.} \cite{Ryzhov_closed, Ryzh_ac_sing}, for the case of general non-selfadjoint operators), whose proof is found in Appendix, is an alternative non-model characterisation of the linear sets
$\widetilde{N}_\pm^\varkappa$.
%The proof is found in the Appendix.
\begin{theorem}
  \label{lem:on-smooth-vectors-other-form}
The sets $\widetilde{N}_\pm^\varkappa$ are described as follows:
\begin{equation}
\widetilde{N}_\pm^\varkappa=\bigl\{u\in\cH: \chi_\varkappa^{\mp}\alpha\Gamma_0(A_{\varkappa}-z
I)^{-1}u\in H^2_\pm(E)\bigr\}\,.
\label{N_tilde_Characterisation}
\end{equation}
\end{theorem}

\begin{corollary}
\label{new_remark}
The right-hand side of (\ref{N_tilde_Characterisation}) coincides with $\{u\in\cH: \alpha\Gamma_0(A_{\varkappa}-zI)^{-1}u\in H^2_\pm(E)\},$ and therefore equivalently one has
\begin{equation}
\widetilde{N}_\pm^\varkappa=\{u\in\cH: \alpha\Gamma_0(A_{\varkappa}-z
I)^{-1}u\in H^2_\pm(E)\}.
\label{Ntilde_form}
\end{equation}
\end{corollary}
\begin{proof}
Indeed, if $\alpha\Gamma_0(A_{\varkappa}-z
I)^{-1}u\in H^2_+(E)$ then clearly $\chi_\varkappa^-\alpha\Gamma_0(A_{\varkappa}-z
I)^{-1}u\in H^2_+(E).$ Conversely,
%for the characteristic function ({\it cf.} (\ref{S_definition}))
%\[
%S(z)=I\eval{E}+\I\alpha\bigl(B_{\I I}^*-M(z)\bigr)^{-1}\alpha\eval{E},\ \ \ \ \ \ z\in\complex_+,
%\]
we write
\begin{align}
S(z)\chi_\varkappa^-\alpha\Gamma_0(A_{\varkappa}-zI)^{-1}u &= (S(z)\chi_\varkappa^-+\chi_\varkappa^+)\alpha\Gamma_0(A_{\varkappa}-zI)^{-1}u-\chi_\varkappa^+\alpha\Gamma_0(A_{\varkappa}-zI)^{-1}u\label{first_line}\\[0.4em]
&= \widehat{\Theta}_\varkappa(z)\alpha\Gamma_0(A_{\varkappa}-zI)^{-1}u-\chi_\varkappa^+\alpha\Gamma_0(A_{\varkappa}-zI)^{-1}u\label{penultimate_line}\\[0.4em]
&=\alpha\Gamma_0(A^*_{iI}-zI)^{-1}u-\chi_\varkappa^+\alpha\Gamma_0(A_{\varkappa}-zI)^{-1}u,\label{last_line}
\end{align}
where $S(z)\chi_\varkappa^-+\chi_\varkappa^+=(S(z)-I)\chi_\varkappa^-+I=\widehat{\Theta}_\varkappa(z),$ see
(\ref{eq:alternate-for-theta-hat}), and in (\ref{penultimate_line})--(\ref{last_line}) we use the part (iii) of Lemma \ref{lem:ryzhov-formulae}.

Further, as we noted in (\ref{eq:naboko-thm-1}), one has $\alpha\Gamma_0(A^*_{iI}-zI)^{-1}u\in H^2_+(E),$ and
since $S$ is an analytic contraction in ${\mathbb C}_+$ the function  $S(z)\chi_\varkappa^-\alpha\Gamma_0(A_{\varkappa}-zI)^{-1}u,$ $z\in{\mathbb C}_+,$ is an element of $H^2_+(E)$ as long as $\chi_\varkappa^-\alpha\Gamma_0(A_{\varkappa}-zI)^{-1}u\in H^2_+(E).$ Recalling (\ref{first_line}), (\ref{last_line}), we conclude that $\chi_\varkappa^+\alpha\Gamma_0(A_{\varkappa}-zI)^{-1}u\in H^2_+(E)$ and therefore
\[
\chi_\varkappa^+\alpha\Gamma_0(A_{\varkappa}-zI)^{-1}u+\chi_\varkappa^-\alpha\Gamma_0(A_{\varkappa}-zI)^{-1}u=\alpha\Gamma_0(A_{\varkappa}-zI)^{-1}u\in H^2_+(E),
\]
as required.

The equality
\[
\bigl\{u\in\cH: \chi_\varkappa^+\alpha\Gamma_0(A_{\varkappa}-zI)^{-1}u\in H^2_-(E)\bigr\}=\bigl\{u\in\cH: \alpha\Gamma_0(A_{\varkappa}-zI)^{-1}u\in H^2_-(E)\bigr\}
\]
is shown in a similar way.
\end{proof}

%\begin{remark}
%\label{lem:N-in-functional-space}
The above corollary together with Theorem \ref{thm:on-smooth-vectors-a.c.equality} motivates generalising the notion of the absolutely continuous subspace $\cH_{\rm ac}(A_{\varkappa})$ to the case of non-selfadjoint extensions $A_\varkappa$ of a symmetric operator $A,$ by identifying it with the set $N^\varkappa_{\rm e}.$ This generalisation follows in the footsteps of the corresponding definition by Naboko \cite{MR573902} in the case of additive perturbations (see also \cite{Ryzhov_closed, Ryzh_ac_sing} for the general case). In particular, an argument similar to \cite[Corollary 1]{MR573902} shows that for the functional model image of $\tilde{N}^\varkappa_{\rm e}$ the following representation holds:
\begin{align}
&\Phi\widetilde{N}^\varkappa_{\rm e}=\biggl\{P_K\binom{\widetilde{g}}{g}\in\mathfrak{H}:\nonumber
\\
&\binom{\widetilde{g}}{g}\in\mathfrak{H}\ {\rm satisfies}\
  \Phi(A_{\varkappa}-z I)^{-1}\Phi^*P_K\binom{\widetilde{g}}{g}=
P_K\frac{1}{\cdot-z}\binom{\widetilde{g}}{g}\ \ \ \forall\,z\in{\mathbb C}_-\cup{\mathbb C}_+\biggr\}.
\label{New_Representation}
\end{align}
(Note that the inclusion of the right-hand side of (\ref{New_Representation}) into $\Phi\tilde{N}^\varkappa_{\rm e}$ follows immediately from Theorem \ref{lem:similar-to-naboko-thm-4}.)
Further,
%for $\widetilde{N}_{\rm e}^\varkappa$
we arrive at an equivalent description:
  \begin{equation}
    \Phi \widetilde{N}_{\rm e}^\varkappa=\left\{P_K\binom{\widetilde{g}}{g}: \binom{\widetilde{g}}{g}\in\mathfrak{H}\ {\rm satisfies}\
 \chi_\varkappa^+(\widetilde{g}+S^*g)+\chi_\varkappa^-(S\widetilde{g}+g)=0\right\}\,.
 \label{lyubimaya_formula}
  \end{equation}
%\end{remark}

\begin{definition}
\label{abs_cont_subspace}
%generalising the notion of the absolutely continuous subspace
For a symmetric operator $A,$ in the case of a non-selfadjoint extension $A_\varkappa$ the absolutely continuous subspace $\cH_{\rm ac}(A_{\varkappa})$  is defined by the formula $\cH_{\rm ac}(A_{\varkappa})=N^\varkappa_{\rm e}.$

In the case of a self-adjoint extension $A_\varkappa$, we understand $\cH_{\rm ac}(A_{\varkappa})$ in the sense of the classical definition of the absolutely continuous subspace of a self-adjoint operator.
\end{definition}

\section{The relationship between the set of smooth vectors and the absolutely continuous
  subspace in the self-adjoint setting}
\label{sec:inclusion-smooth}

The argument of this section is similar to that of \cite{MR573902},
subject to appropriate modifications in order to account for the fact
that we deal with extensions of symmetric operators rather than
additive perturbations. The same strategy seems to be applicable in the
``mixed'' case that incorporates both extensions and perturbations, which has recently been studied in \cite{MR2418300}.

The following proposition is contained in the proof of \cite[Lemma 5]{MR573902}. For reader's convenience, we provide its proof in Appendix.
\begin{proposition}
  \label{lem:brothers-naboko}
  If the Borel transform of a Borel measure $\mu$
  \begin{equation*}
    \int_\reals\frac{d\mu(s)}{s-z}
  \end{equation*}
is either an element of $H_+^2$ when $z\in\complex_+$ or an element of
$H_-^2$ when $z\in\complex_-$, then $\mu$ is absolutely
continuous with respect to the Lebesgue measure.
\end{proposition}
\begin{lemma}
  \label{lem:singular-projection-sa-part}
Assume that $\varkappa=\varkappa^*,$  ${\rm ker}(\alpha)=\{0\}$ and
%and $\alpha\Gamma_0$ is densely defined and closable.
let $P_S$ be the orthogonal projection onto the singular subspace of $A_{\varkappa}$. Then
following inclusion holds:
\begin{equation*}
  P_S\widetilde{N}_{\rm e}^\varkappa\subset\bigcap_{z\in\complex\setminus\reals}\ran(A-zI)\,.
\end{equation*}
\end{lemma}
\begin{proof}
We first demonstrate the validity of the claim for $\varkappa=0.$

We decompose each smooth vector $u$ (i.e $u\in\widetilde{N}^\varkappa_{\rm e}$) into its projections onto the
absolutely continuous and singular subspaces of $A_0$, that
is, $u=u_{\rm ac}+u_s$, where $u_{\rm ac}\in\cH_{\rm ac}(A_{0})$ and
$u_s\in\cH_s(A_{0})$, so $u_{\rm ac}\perp u_s$ and $u_s\in P_S\widetilde{N}_{\rm e}^\varkappa.$

Consider an arbitrary $w\in{\mathcal K}$ and note that, due to the surjectivity of $\Gamma_1,$ there exists a vector $v\in{\rm dom} (A^*)$ such that $\alpha w=\Gamma_1v,$ and therefore
\begin{align}
%\inner{\alpha\Gamma_0(A_{0}-zI)^{-1}u}{w}_{\mathcal{K}}=
\inner{\Gamma_0(A_{0}-zI)^{-1}u}{\alpha w}_{\mathcal{K}}&=\inner{\Gamma_0(A_{0}-zI)^{-1}u}{\Gamma_1 v}_{\mathcal{K}}\label{first_measure}\\[0.5em]
&=\inner{\Gamma_0(A_{0}-zI)^{-1}u}{\Gamma_1 v}_{\mathcal{K}}-\inner{\Gamma_1(A_{0}-zI)^{-1}u}{\Gamma_0v}_{\mathcal{K}}\label{second_measure}\\[0.5em]
&=\inner{(A_{0}-zI)^{-1}u}{A^*v}_\cH-\inner{A^*(A_{0}-zI)^{-1}u}{v}_\cH\label{third_measure}\\[0.6em]
%\label{third_measure}
%\\
&=\int_\reals\frac{1}{t-z}d\mu_{u, A^*v}(t)-\int_\reals\frac{t}{t-z}d\mu_{u,v}(t)=\int_\reals\frac{1}{t-z}d\hat{\mu}(t).\label{last_measure}
\end{align}
Here %the measure $\mu_{u_s,A^*v}$ is defined as
%\begin{equation}
%\label{eq:mu-def}
 \[
 \mu_{u,A^*v}(\delta):=\inner{E_{A_0}(\delta)u}{A^*v}_\cH, \ \ \ \ \mu_{u,v}(\delta):=\inner{E_{A_0}(\delta)u}{v}_\cH\ \ \ \ \forall\,{\rm Borel}\ \delta\subset\reals,
 \]
%\end{equation}
where $E_{A_0}$ is the spectral resolution of the identity for the operator $A_0,$ and $\hat{\mu}(t):=\mu_{u, A^*v}(t)-t\mu_{u,v}(t).$ Furthermore, the measure $\hat{\mu}$ admits the decomposition into its absolutely continuous and singular parts with respect to the Lebesgue measure. Its singular part is equal to $\mu_{u_s, A^*v}(t)-t\mu_{u_s,v}(t)=:\hat{\mu}_s(t),$ see {\it e.g.} \cite{MR1192782}. The equality (\ref{first_measure})--(\ref{second_measure}) is due to the observation that $\Gamma_1$ vanishes on ${\rm dom}(A_0),$ and the equality (\ref{second_measure})--(\ref{third_measure}) is a consequence of the ``Green formula" (\ref{Green_formula}) and the fact that $A\subset A_0.$

At the same time, it follows from
%Theorem \ref{lem:on-smooth-vectors-other-form}, see also
Corollary \ref{new_remark} that the scalar analytic function $\inner{\Gamma_0(A_{0}-zI)^{-1}u}{\alpha w}_{\mathcal{K}}$ is an element of $H^2_+$ for $z\in{\mathbb C}_+$ and also of $H^2_-$ for $z\in{\mathbb C}_-.$ Therefore, by Proposition \ref{lem:brothers-naboko} we infer from (\ref{first_measure})--(\ref{last_measure}) that  the measure $\hat{\mu}$ is absolutely continuous, which implies that its singular part $\hat{\mu}_s$ is the zero measure.

%%which, in view of the fact that it is also singular, implies that it is the zero measure.

%%\begin{align}
%\inner{\alpha\Gamma_0(A_{0}-zI)^{-1}u}{w}_{\mathcal{K}}=
%%\inner{\Gamma_0(A_{0}-zI)^{-1}u_s}{\alpha w}_{\mathcal{K}}&=\inner{\Gamma_0(A_{0}-zI)^{-1}u_s}{\Gamma_1 v}_{\mathcal{K}}\nonumber\label{first_measure}
%%\\[0.5em]
%%&=\inner{\Gamma_0(A_{0}-zI)^{-1}u_s}{\Gamma_1 v}_{\mathcal{K}}-\inner{\Gamma_1(A_{0}-zI)^{-1}u_s}{\Gamma_0v}_{\mathcal{K}}%\label{second_measure}
%%\nonumber
%%\\[0.5em]
%%&=\inner{(A_{0}-zI)^{-1}u_s}{A^*v}_\cH-\inner{A^*(A_{0}-zI)^{-1}u_s}{v}_\cH\label{third_measure}\\[0.6em]
%\label{third_measure}
%\\
%%&=\int_\reals\frac{1}{t-z}d\mu_{u_s, A^*v}(t)-\int_\reals\frac{t}{t-z}d\mu_{u_s,v}(t)=\int_\reals\frac{1}{t-z}d\hat{\mu}(t).\label{last_measure}
%%\end{align}

 Finally, we invoke (\ref{first_measure})--(\ref{last_measure}) once again, having replaced $u$ by $u_s$ and $\hat{\mu}$ by $\hat{\mu}_s,$ and conclude that
\begin{equation}
\inner{\Gamma_0(A_{0}-zI)^{-1}u_s}{\alpha w}_{\mathcal{K}}=0\ \ \ \ \forall\,z\in{\mathbb C}\setminus{\mathbb R}.
\label{zero_id}
\end{equation}

Now, by virtue of the facts that $w\in{\mathcal K}$ in (\ref{zero_id}) is arbitrary and ${\rm ker}(\alpha)=\{0\},$
it follows that
$\Gamma_0(A_{0}-zI)^{-1}u_s=0,$ and since $(A_{0}-zI)^{-1}u_s\in{\rm dom}(A_0)$ and therefore
$\Gamma_1(A_{0}-zI)^{-1}u_s=0$ automatically, we obtain
 %$(A_{0}-zI)^{-1}u=0.$ Clearly,
 $(A_{0}-zI)^{-1}u_s\in{\rm dom}(A).$ Finally, since $A_0\supset A,$ we conclude that $u_s\in{\rm ran}(A-zI)$ for all $z\in{\mathbb C}\setminus{\mathbb R},$ as claimed.

In order to treat the case of an arbitrary $\varkappa\in{\mathcal B}({\mathcal K})$ such that $\varkappa=\varkappa^*,$ we define ``shifted'' boundary operators $\widehat{\Gamma}_0:=\Gamma_0,$ $\widehat{\Gamma}_1:=\Gamma_1-B_\varkappa\Gamma_0.$ Notice that ({\it cf.} (\ref{eq:extension-by-operator}))
\[
{\rm dom}(A_\varkappa)=\{u\in\cH: \Gamma_1u=B_\varkappa\Gamma_0u\}=\{u\in\cH: \widehat{\Gamma}_1u=0\},
\]
{\it i.e.} the operator $A_\varkappa$ plays the r\^{o}le of the operator $A_0$ in the triple $({\mathcal K}, \widehat{\Gamma}_0, \widehat{\Gamma}_1).$ Further, note that the change of the triple results in a change of the operator that needs to play the r\^{o}le of $A_{iI},$ the dissipative extension used to construct the functional model, which in terms of the ``old'' triple $({\mathcal K, \Gamma_0, \Gamma_1})$ should be the extension $A_B$ with $B=\alpha(i+\varkappa)\alpha/2.$ Repeating the above argument in this new functional model and bearing in mind that the characterisation of $\widetilde{N}^\varkappa_{\rm e}$ in Corollary \ref{new_remark} holds for all $\varkappa,$ yields the stated result.
\end{proof}
An immediate consequence of this result and the
criterion of complete non-selfadjoint\-ness
(\ref{eq:completely-non-selfadjoint-criterion}) is the following assertion.
\begin{corollary}
  \label{cor:on-smooth-vectors-a.c.contention}
  Let $\varkappa$ and $\alpha$ be as in the preceding lemma. If $A$
  is completely non-selfadjoint, then
  \begin{equation*}
    \widetilde{N}_{\rm e}^\varkappa\subset \cH_{\rm ac}(A_{\varkappa})\,.
  \end{equation*}
\end{corollary}
We now proceed to the proof of the opposite inclusion.
\begin{lemma}[Modified Rosenblum lemma, {\it cf.} \cite{MR1307384}]
  \label{lemma-rosenblume}
  Let $\beta$ be a self-adjoint operator in a Hilbert space
  $\cH_1$. Suppose that the operator $T$, defined on $\dom(\beta)$ and
  taking values in a Hilbert space $\cH_2$, is such that $T(\beta-z_0
  I)^{-1}$ is a Hilbert-Schmidt operator for some
  $z_0\in\rho(\beta)$. Then there exists a set $\mathcal{D}$, dense in
  $\cH_{\rm ac}(\beta)$, such that
  \begin{equation*}
    \int_\reals\norm{T\exp(-i\beta t)u}^2dt<\infty
  \end{equation*}
 for all $u\in\mathcal{D}$.
\end{lemma}
\begin{proof}
 Let $x\in\reals$ and $\epsilon>0$. By Hilbert's first identity
 \begin{equation*}
    T(\beta-(x+\I\epsilon)
    I)^{-1}=((x+\I\epsilon)-z_0)T(\beta-z_0
    I)^{-1}(\beta-(x+\I\epsilon)
    I)^{-1} + T(\beta-z_0 I)^{-1}
  \end{equation*}
  Consider the first term on the right-hand side of this last
  equation. By \cite{MR1036844}, for every $f$ in
  $\cH_1$ the limit
\begin{equation*}
  \lim_{\epsilon\to 0}T(\beta-z_0
    I)^{-1}(\beta-(x+\I\epsilon)
    I)^{-1}f
\end{equation*}
exists for almost all $x\in\reals$ (the convergence set
actually depends on $f$). It follows that the limit
\begin{equation*}
  \lim_{\epsilon\to 0}T\left((\beta-(x+\I\epsilon)
    I)^{-1}-(\beta-(x-\I\epsilon)
    I)^{-1} \right)f=:F(x)
\end{equation*}
exists for all $f\in\cH_1$ and almost all $x\in\reals$.

Now, define the set
\begin{equation*}
  \mathcal{X}(n):=\bigl\{x\in\reals:\abs{x}<n, \norm{F(x)}<n\bigr\}
\end{equation*}
If $E_\beta$ denotes the spectral measure of the operator $\beta$,
then the set
\begin{equation*}
  \mathcal{D}:=\bigcup_{n\in\nats}E_\beta\bigl(\mathcal{X}(n)\bigr)\cH_{\rm ac}(\beta)
\end{equation*}
is dense in $\cH_{\rm ac}(\beta)$. Consider an orthonormal basis
$\{\phi_k\}$ in $\cH_2$ and an arbitrary element $f\in\mathcal{D}$,
then, for all $k$,
\begin{align*}
  \inner{T\exp(-i\beta t)f}{\phi_k}&=\int_{\mathcal{X}(n)}e^{-\I x
    t}\frac{d}{dx}\inner{E_\beta(x)f}{T^*\phi_k}dx\\ &=
  \int_{\mathcal{X}(n)}e^{-\I x t}\inner{F(x)}{T^*\phi_k}dx\,,
\end{align*}
where in the last equality we have used the fact that by the spectral theorem
\begin{equation*}
  \lim_{\epsilon\to 0}\inner{\left((\beta-(x+\I\epsilon)
    I)^{-1}-(\beta-(x-\I\epsilon)
    I)^{-1} \right)f}{\phi}=\frac{d}{dx}\inner{E_\beta(x)f}{\phi}
    %\eval{t=x}
\end{equation*}
for all $f\in\cH_{\rm ac}(\beta)$ and for all $\phi\in\cH_1$.

By the Parseval identity one has
\begin{equation*}
  \int_\reals\abs{\inner{T\exp(-i\beta t)f}{\phi_k}}^2dt=
2\pi\int_{\mathcal{X}(n)}\abs{\inner{F(x)}{\phi_k}}^2dx
\end{equation*}
for all $k$, which immediately implies that
\begin{equation*}
  \int_\reals\norm{T\exp(-i\beta t)u}^2dt=
2\pi\int_{\mathcal{X}(n)}\norm{F(x)}^2dx\le
4\pi n^3<+\infty\,.
\end{equation*}
\end{proof}
Combining the above statements yields the following result.
\begin{theorem}
  \label{thm:on-smooth-vectors-a.c.equality}
  %Let $\varkappa$ and $\alpha$ be as in
  Assume that $\varkappa=\varkappa^*,$  ${\rm ker}(\alpha)=\{0\}$
  %Lemma~\ref{lem:singular-projection-sa-part}
  and let  $\alpha\Gamma_0(A_{\varkappa}-z I)^{-1}$ be a Hilbert-Schmidt
  operator for at least one point $z\in\rho(A_\varkappa)$. If
  $A$ is completely non-selfadjoint, then our definition of the absolutely continuous subspace is
  equivalent to the classical definition of the absolutely continuous subspace of a self-adjoint operator, {\it i.e.}
  \begin{equation*}
    N_{\rm e}^\varkappa = \cH_{\rm ac}(A_{\varkappa})\,.
  \end{equation*}
\end{theorem}
\begin{proof}
By applying the Fourier transform to the functions
$\mathbbm{1}_\pm(t)\alpha\Gamma_0e^{iA_{\varkappa}t}e^{\mp \epsilon
  t}u$, $t\in{\mathbb
  R},$
where $\mathbbm{1}_\pm$ is the
characteristic function of $\reals_\pm$ and $\epsilon>0$ is
arbitrarily small, one obtains
\begin{equation*}
  \norm{\alpha\Gamma_0(A_{\varkappa}-z I)^{-1}u}_{H_-^2}^2
+\norm{\alpha\Gamma_0(A_{\varkappa}-z I)^{-1}u}_{H_+^2}^2
=2\pi\int_\reals\norm{\alpha\Gamma_0\exp(iA_{\varkappa}t)u}^2dt\,
\end{equation*}
which by Lemma~\ref{lemma-rosenblume} is finite for all $u$ in a
dense subset of $\cH_{\rm ac}(A_{\varkappa})$. Hence, in view of Corollary \ref{new_remark}
%\ref{lem:on-smooth-vectors-other-form}
and performing closure,
% in $\cH_{\rm ac}(A_{\varkappa}),$
one has $\cH_{\rm ac}(A_{\varkappa})\subset N^\varkappa_{\rm e}.$ Taking into
account Corollary~\ref{cor:on-smooth-vectors-a.c.contention} completes the proof.
\end{proof}

\begin{remark}
Alternative conditions, which are less restrictive in general, that guarantee the validity of the assertion of Theorem
\ref{thm:on-smooth-vectors-a.c.equality} can be obtained along the lines of \cite{MR1252228}.
\end{remark}

\section{Wave and scattering operators}
\label{sec:wave-operators}
The results of the preceding sections allow us to calculate the wave
operators for any pair $A_{\varkappa_1},A_{\varkappa_2}$, where
$A_{\varkappa_1}$ and $A_{\varkappa_2}$ are operators in the class
introduced in Section~\ref{sec:boundary-triples}, under the additional
assumption that the operator $\alpha$ (see (\ref{eq:b-kappa-def})) has a trivial kernel.
For simplicity,
%and bearing in mind the
%application of the abstract construction to the problem described in
%Sections \ref{sec:quantum-graphs} and \ref{sec:inverse-scattering},
in what follows we set $\varkappa_2=0$ and write
$\varkappa$ instead of $\varkappa_1$. Note that $A_0$ is a self-adjoint
operator, which is convenient for presentation purposes.

We begin by establishing the model representation for the function
$\exp(iA_\varkappa t)$, $t\in\reals$, of the
operator $A_\varkappa,$ evaluated on the set of smooth vectors $\widetilde{N}_{\rm e}^\varkappa.$
\begin{proposition}(\cite[Prop.\,2]{MR573902})
  \label{prop:exp-funtion}
  For all $t\in\reals$ and all $\binom{\widetilde{g}}{g}$ such that
  $\Phi^*P_K\binom{\widetilde{g}}{g}\in\widetilde{N}_{\rm e}^\varkappa$ one has
  \begin{equation*}
    \Phi\exp(iA_\varkappa t)\Phi^*
  P_K\binom{\widetilde{g}}{g}=
P_K\exp(ikt)\binom{\widetilde{g}}{g}.
  \end{equation*}
\end{proposition}
\begin{proof}
  We use the definition
  \begin{equation*}
    \exp(iA_\varkappa t):=
 \slim_{n\to+\infty}\left(I-\frac{iA_\varkappa t}{n}\right)^{-n},\ \ \ t\in{\mathbb R},
  \end{equation*}
  giving in general an unbounded operator (see \cite{MR0407617}). Due
  to Theorem~\ref{lem:similar-to-naboko-thm-4}, if
  $\binom{\widetilde{g}}{g}\in\mathfrak{N}^\varkappa_+\cap\mathfrak{N}_-^\varkappa$, {\it i.e.}
  $\Phi^*P_K\binom{\widetilde{g}}{g}\in\widetilde{N}_{\rm e}^\varkappa$,
  then
\begin{equation*}
  \left(I-\frac{iA_\varkappa
      t}{n}\right)^{-n}\Phi^*P_K\binom{\widetilde{g}}{g}=
 \Phi^*P_K\left(1-\frac{ikt}{n}\right)^{-n} \binom{\widetilde{g}}{g},\ \ \ \ t\in{\mathbb R}.
\end{equation*}
Thus, to complete the proof it remains to show that
\begin{equation*}
  \norm{\left(\exp(ikt)-\left(1-\frac{ikt}{n}\right)^{-n}\right)
\binom{\widetilde{g}}{g}}_{\mathfrak H}\convergesto{n}{\infty}0,\ \ \ t\in{\mathbb R},
\end{equation*}
which follows directly from Lebesgue's dominated convergence theorem.
\end{proof}
\begin{proposition}(\cite[Section 4]{MR573902})
  \label{prop:pre-wave-op}
  If $\Phi^*P_K\binom{\widetilde{g}}{g}\in\widetilde{N}_{\rm e}^\varkappa$
  and $\Phi^*P_K\binom{\widehat{g}}{g}\in\widetilde{N}_{\rm e}^0$
  (with the same element\footnote{Despite the fact that
    $\binom{\widetilde{g}}{g}\in\mathfrak{H}$ is nothing but a symbol,
    still $\widetilde{g}$ and $g$ can be identified with vectors in
    certain $L^2(E)$ spaces with  operators ``weights'', see
    details below in Section~\ref{sec:spectral-repr-ac}. Further, we recall that even then
    %it must be noted that the assumption of the Proposition is not trivial since if
    for $\binom{\widetilde{g}}{g}\in\mathfrak{H}$, the components $\widetilde{g}$ and $g$ are not, in general, \emph{independent} of each other.} $g$), then
\begin{equation*}
  \norm{\exp(-iA_\varkappa t)\Phi^*P_K\binom{\widetilde{g}}{g}-\exp(-iA_0
  t)\Phi^* P_K\binom{\widehat{g}}{g}}_{\mathfrak H}\convergesto{t}{-\infty}0.
\end{equation*}
\end{proposition}
\begin{proof}
  %%Since $\Phi^*P_K\binom{\widehat{g}}{g}$ is in $\widetilde{N}_{\rm e}^0$, one has
 % \begin{equation*}
%%$    \widehat{g}+S^*g+S\widehat{g}+g=0.$
 % \end{equation*}
Clearly, $\widetilde{g}-\widehat{g}\in L^2(E)$
since $\binom{\widetilde{g}-\widehat{g}}{0}\in\mathfrak{H}.$ Therefore, for all $t\in{\mathbb R},$ we obtain
\begin{align*}
  \norm{\exp(-iA_\varkappa t)\Phi^*P_K\binom{\widetilde{g}}{g}-\exp(-iA_0
  t)\Phi^*P_K\binom{\widehat{g}}{g}}_{\mathfrak H}&=\norm{P_K
  e^{-it\cdot}\binom{\widetilde{g}}{g}-P_K
e^{-it\cdot}\binom{\widehat{g}}{g}}_{\mathfrak H}\\[0.5em]
&=\norm{P_K\binom{e^{-it\cdot}(\widetilde{g}-\widehat{g})}{0}}_{\mathfrak H}\\[0.5em]
%&=\norm{P_K\binom{P_-e^{-it\cdot}(\widetilde{g}-\widehat{g})}{0}}\\
&\le \norm{P_-e^{-it\cdot}(\widetilde{g}-\widehat{g})}_{L^2(E)}\,.
\end{align*}
where in the inequality we use the fact that
\begin{equation*}
  \norm{P_K\binom{\check{g}}{0}}_{\mathfrak H}^2
  %&=&\norm{P_K\binom{P_-\check{g}}{0}}^2 \label{thefact}\\[0.5em]
  =\int_\reals
\left(\bigl\Vert P_-\check{g}(s)\bigr\Vert_{E}^2-\bigl\Vert P_-S(s)\check{g}(s)\bigr\Vert_{E}^2\right)ds\quad\quad\forall\binom{\check{g}}{0}\in\mathfrak{H}.
%\nonumber
\end{equation*}
%%The equality (\ref{thefact}) follows from the identity
%The third equality above follows from the observation that, for all
%$\binom{\check{g}}{0}\in\mathfrak{H}$, one has
%%\begin{equation*}
%%  P_K\binom{\check{g}}{0}-P_K\binom{P_-\check{g}}{0}=\binom{0}{P_-SP_+\check{g}}=0
%%  \quad\quad\forall\binom{\check{g}}{0}\in\mathfrak{H},
%%\end{equation*}
%%since $S(z),$ $z\in{\mathbb C}_+,$ is an analytic contraction in the upper half-plane.

 %It can also be verified that
Finally, since $\exp(-it\cdot)\in H_+^\infty$ for
$t\ge 0,$ the convergence (see {\it e.g.} \cite{MR1669574})
\begin{equation*}
  \norm{P_-e^{-it\cdot}(\widetilde{g}-\widehat{g})}_{L^2(E)}^2=
\int_{-\infty}^t\norm{ \mathscr{F}(\widetilde{g}-\widehat{g})(\tau)}_E^2d\tau
\convergesto{t}{-\infty}0
\end{equation*}
holds, where $ \mathscr{F}(\widetilde{g}-\widehat{g})$ stands for the Fourier
transform of the function $\widetilde{g}-\widehat{g}$.
\end{proof}
It follows from Proposition \ref{prop:pre-wave-op}  that whenever $\Phi^*P_K\binom{\widetilde{g}}{g}\in\widetilde{N}_{\rm e}^\varkappa$
and $\Phi^*P_K\binom{\widehat{g}}{g}\in\widetilde{N}_{\rm e}^0$ (with the same second component $g$), formally one has
\begin{align*}
\lim_{t\to-\infty}e^{iA_0t}e^{-iA_\varkappa t}\Phi^*P_K\binom{\widetilde{g}}{g}&=\Phi^*P_K\binom{\widehat{g}}{g}\\
&=\Phi^*P_K\binom{-(I+S)^{-1}(I+S^*)g}{g}\,,
\end{align*}
where in the last equality we use the inclusion $\Phi^*P_K\binom{\widehat{g}}{g}\in\widetilde{N}_{\rm e}^0,$ which by
%Remark~\ref{lem:N-in-functional-space} ({\it cf.}
 (\ref{lyubimaya_formula}) yields $\widehat{g}+S^*g+S\widehat{g}+g=0.$
%and the fact that
%$\Phi^*P_K\binom{\widehat{g}}{g}\in\widetilde{N}_{\rm e}^0.$
In view of the classical definition of the wave operator of a pair of self-adjoint operators, see {\it e.g.} \cite{MR0407617},
%the standard notation
\begin{equation*}
  W_\pm(A_0,A_\varkappa):=\slim_{t\to\pm\infty}e^{iA_0t}e^{-iA_\varkappa t}P_{\rm ac}^\varkappa,
\end{equation*}
where $P_{\rm ac}^\varkappa$ is the projection onto the absolutely continuous subspace of $A^\varkappa,$ we obtain that, at least formally, for $\Phi^*P_K\binom{\widetilde{g}}{g}\in\widetilde{N}_{\rm e}^\varkappa$ one has
\begin{equation}
  \label{eq:formula-0-kappaw-}
  W_-(A_0,A_\varkappa)\Phi^*P_K\binom{\widetilde{g}}{g}=\Phi^*P_K\binom{-(I+S)^{-1}(I+S^*)g}{g}\,.
\end{equation}

By an argument similar to that of Proposition \ref{prop:pre-wave-op} ({\it i.e.} considering the case $t\to+\infty$), one also obtains
\begin{equation}
\label{eq:formula-0-kappaw+}
W_+(A_0,A_\varkappa)\Phi^*P_K\binom{\widetilde{g}}{g}=
\lim_{t\to+\infty}e^{iA_0t}e^{-iA_\varkappa t}\Phi^*P_K\binom{\widetilde{g}}{g}
=\Phi^*P_K\binom{\widetilde{g}}{-(I+S^*)^{-1}(I+S)\widetilde{g}}
\end{equation}
again for $\Phi^*P_K\binom{\widetilde{g}}{g}\in\widetilde{N}_{\rm e}^\varkappa$.

%Further, in view of the limits
Further, the definition of the wave operators $W_\pm(A_\varkappa,A_0)$
\begin{equation*}
  \norm{e^{-iA_\varkappa
      t}W_\pm(A_\varkappa,A_0)\Phi^*P_K\binom{\widetilde{g}}{g}-e^{-iA_0t}\Phi^*P_K\binom{\widetilde{g}}{g}
}_{\mathfrak H}\convergesto{t}{\pm\infty}0
\end{equation*}
yields, for all $\Phi^*P_K\binom{\widetilde{g}}{g}\in\widetilde{N}_{\rm e}^0,$
\begin{equation}
  \label{eq:formula-kappa-0w-}
  W_-(A_\varkappa,A_0)\Phi^*P_K\binom{\widetilde{g}}{g}=\Phi^*P_K\binom{-(I+\chi_\varkappa^-(S-I))^{-1}(I+\chi_\varkappa^+(S^*-I))g}{g}
\end{equation}
and
\begin{equation}
 \label{eq:formula-kappa-0w+}
  W_+(A_\varkappa,A_0)\Phi^*P_K\binom{\widetilde{g}}{g}=\Phi^*P_K
\binom{\widetilde{g}}{-(I+\chi_\varkappa^+(S^*-I))^{-1}
(I+\chi_\varkappa^-(S-I))\widetilde{g}},
\end{equation}
where we have used the fact that  $\Phi^*P_K\binom{\widetilde{g}}{g}\in\widetilde{N}_{\rm e}^\varkappa$ and the corresponding criterion provided by
%Remark \ref{lem:N-in-functional-space}, {\it cf.}
(\ref{lyubimaya_formula}).
%%%%%%%%%%%%%%%%%%%%

%Having obtained formulae
In order to rigorously justify the above formal argument, {\it i.e.} in order to prove the existence and completeness of the wave operators, one needs to first show that
the right-hand sides of the formulae \eqref{eq:formula-0-kappaw-}--\eqref{eq:formula-kappa-0w+} make sense on dense subsets of the corresponding absolutely continuous subspaces.
%and then to demonstrate that the linear operators defined by these formulae are bounded.
Noting that  \eqref{eq:formula-0-kappaw-}--\eqref{eq:formula-kappa-0w+} have the form identical to the expressions for wave operators derived in \cite[Section 4]{MR573902}, \cite {MR1252228}, the remaining part of this justification is a modification of the argument of \cite {MR1252228}, as follows.
%In particular argument is based on

%one still
%has to prove the existence and completeness of the wave operators. To
%this end, we recur to scalar multiples whose definition is given next.

Let $S(z)-I$ be of the class $\mathfrak {S}_\infty(\overline{ \mathbb {C}}_+),$ {\it i.e.} a compact analytic operator function in the upper half-plane up to the real line. Then so is $(S(z)-I)/2$, which is also uniformly bounded in the upper half-plane along with $S(z)$. We next use the result of \cite[Theorem 3]{MR1252228} about the non-tangential boundedness of operators of the form $(I+T(z))^{-1}$ for $T(z)$ compact up to the real line. We infer that, provided $(I+(S(z_0)-I)/2)^{-1}$ exists for some $z_0\in \mathbb C_+$ (and hence, see \cite{Brodski}, everywhere in $\mathbb C_+$ except for a countable set of points accumulating only to the real line), one has non-tangential boundedness of $(I+(S(z)-I)/2)^{-1}$, and therefore also of $(I+S(z))^{-1}$, for almost all points of the real line.

On the other hand, the latter inverse can be computed in $\mathbb C_+$:
\begin{equation}\label{id1}
\bigl(I+S(z)\bigr)^{-1}=\frac{1}{2}\bigl(I+ i \alpha M(z)^{-1}\alpha/2\bigr).
\end{equation}
Indeed, one has
\begin{multline*}
\bigl(I+ i \alpha M(z)^{-1}\alpha/2\bigr)(I+S(z))\\[0.4em] =2I + i \alpha M(z)^{-1} \alpha +
i \alpha\bigl(B_{iI}^*-M(z)\bigr)^{-1}\alpha -i\alpha M(z)^{-1} B_{iI}^*\bigl(B_{iI}^*-M(z)\bigr)^{-1}\alpha=2I
\end{multline*}
and the second similar identity for the multiplication in the reverse order proves the claim.

It follows from (\ref{id1}) and the analytic properties of $M(z)$ that the inverse $(I+S(z))^{-1}$ exists everywhere in the upper half-plane. Thus, Theorem 3 of \cite{MR1252228}  is indeed applicable, which yields that
$(I+S(z))^{-1}$ is $\mathbb R$-a.e. nontangentially bounded and, by the operator generalisation of the Calderon theorem (see \cite{Calderon}), which was extended to the operator context in \cite[Theorem 1]{MR1252228}, it admits measurable non-tangential limits in the strong operator topology almost everywhere on $\mathbb R$. As it is easily seen, these limits must then coincide with $(I+ S(k))^{-1}$ for almost all $k\in \mathbb R$.

The same argument obviously applies to $(I +S^*(\bar{z}))^{-1}$ for $z\in \mathbb C_-,$ where the invertibility follows from the identity
\begin{equation}\label{id2}
\bigl(I+S^*(\bar z)\bigr)^{-1}=\frac{1}{2}\bigl(I - i \alpha M(z)^{-1}\alpha/2\bigr)
\end{equation}
obtained exactly as \eqref{id1}, by taking into account analytic properties of $M(z)$.

Finally, the identities
\begin{equation}\label{id3}
(I+\chi_\varkappa^-(S(z)-I))^{-1}=I - i \chi_\varkappa^-\alpha (B_\varkappa-M(z))^{-1}\alpha
\end{equation}
for $z\in \mathbb C_+$ and
\begin{equation}\label{id4}
(I+\chi_\varkappa^+(S^*(\bar{z})-I))^{-1}=I+i \chi_\varkappa^+\alpha (B_\varkappa-M(z))^{-1}\alpha
\end{equation}
for $z\in \mathbb C_-$ are used, again by an application of Theorem 3 of \cite{MR1252228}, to ascertain the existence of bounded $(I+\chi_\varkappa^-(S(k)-I))^{-1}$ and $(I+\chi_\varkappa^+(S^*(k)-I))^{-1}$ almost everywhere on $\mathbb R,$ provided that the operator $A_\varkappa$ has at least one regular point in each half-plane of the complex plane, see Proposition \ref{prop:weyl-function-properties}.  Under the assumptions on $S$ specified above, this latter condition immediately implies that the non-real spectrum of $A_\varkappa$ is countable and accumulates to $\mathbb R$ only. (Nevertheless, it could still accumulate to all points of the real line simultaneously.)

The presented argument allows one to verify the correctness of the formulae (\ref{eq:formula-0-kappaw-})--(\ref{eq:formula-kappa-0w+}) for the wave operators. Indeed, for the first of them one considers $\mathbbm{1}_n(k),$ the indicator of the set $\{k\in \mathbb R: \|(I+S(k))^{-1}\|\leq n\}.
$
Clearly, $\mathbbm{1}_n(k)\to 1$ as $n\to \infty$ for almost all $k\in \mathbb R$. Next, suppose that $P_K(\tilde g, g)\in \tilde N_{\rm e}^\varkappa$. Then $P_K\mathbbm{1}_n(\tilde g, g)$ is also a smooth vector and
$$
\binom{-(I+S)^{-1}\mathbbm{1}_n (I+S^*) g}{\mathbbm{1}_n g}\in \mathfrak H.
$$
Indeed, for any $(\tilde g,g)\in \mathfrak H$ one has
\begin{multline*}
\binom{-\mathbbm{1}_n(1+S)^{-1}(1+S^*)g}{\mathbbm{1}_n g}-\binom{\mathbbm{1}_n \tilde g}{\mathbbm{1}_n g}\\
=\binom{-\mathbbm{1}_n (1+S)^{-1}[(\tilde g + S^* g)+(S\tilde g+g)]}{0}\in
\binom{L^2(E)}{0}\in \mathfrak H,
\end{multline*}
whereas the inclusion in the set of smooth vectors follows directly from (\ref{lyubimaya_formula}).
%Remark \ref{lem:N-in-functional-space}.
It follows, by the Lebesgue dominated convergence theorem, that the set of vectors $P_K\mathbbm{1}_n(\tilde g, g)$ is dense in $N_{\rm e}^\varkappa$. The remaining three wave operators are treated in a similar way. Finally, the density of the range of the four wave operators follows from the density of their domains, by a standard inversion argument, see {\it e.g.} \cite{MR1180965}.
% completeness of the four wave operators follow$W_-(A_0,A_\varkappa)$ would follow from the
%density of the domain of $W_+(A_\varkappa,A_0)$.

We have thus proved the following theorem.
\begin{theorem}
  \label{thm:existence-completeness-wave-operators}
  Let $A$ be a closed, symmetric, completely nonselfadjoint operator
  with equal deficiency indices and consider its extension $A_\varkappa,$ as described in Section \ref{sec:boundary-triples}, under the assumptions that ${\rm ker}(\alpha)=\{0\}$ (see (\ref{eq:b-kappa-def})) and that $A_\varkappa$ has at least one regular point in ${\mathbb C}_+$ and in ${\mathbb C}_-.$ If $S-I\in\mathfrak {S}_\infty(\overline{ \mathbb {C}}_+),$
  %either $\im\varkappa=o$
  %or $(\im\kappa)^2=I\eval{E}$,
  then the wave operators
  $W_\pm(A_0,A_\varkappa)$ and $W_\pm(A_\varkappa,A_0)$  exist on dense sets in $N_{\rm e}^\varkappa$ and
  $\mathcal{H}_{\rm ac}(A_0)$, respectively, and are given by the formulae
  (\ref{eq:formula-0-kappaw-})--(\ref{eq:formula-kappa-0w+}). The
  ranges of $W_\pm(A_0,A_\varkappa)$ and  $W_\pm(A_\varkappa,A_0)$ are dense in $\mathcal{H}_{\rm ac}(A_0)$ and $N_{\rm e}^\varkappa,$ respectively.\footnote{In the case when $A_\varkappa$ is self-adjoint, or, in general, the named wave operators are bounded, the claims of the theorem are equivalent (by the classical Banach-Steinhaus theorem) to the statement of the existence and completeness of the wave operators for the pair $A_0, A_\varkappa.$ Sufficient conditions of boundedness of these wave operators are contained in {\it e.g.} \cite[Section 4]{MR573902}, \cite {MR1252228} and references therein.}
\end{theorem}

\begin{remark}
\label{long_remark}
1. The identities \eqref{id1}--\eqref{id2} can be used to replace the condition $S(z)-I\in \mathfrak S_\infty(\overline{ \mathbb {C}}_+)$ by the following equivalent condition: $\alpha M(z)^{-1}\alpha$ is nontangentially bounded almost everywhere on the real line, and $\alpha M(z)^{-1}\alpha\in \mathfrak S_\infty(\overline{ \mathbb {C}}_+)$ for $\Im z\geq 0$. In order to do so, one notes that $(I+T)^{-1}-I=-(I+T)^{-1}T\in\mathfrak S_\infty(\overline{ \mathbb {C}}_+)$ as long as $T\in\mathfrak S_\infty(\overline{ \mathbb {C}}_+)$ and $(I+T)^{-1}$ is bounded.

2. The latter condition is satisfied \cite{Krein}, as long as the scalar function $\|\alpha M(z)^{-1}\alpha\|_{\mathfrak S_p}$ is nontangentially bounded almost everywhere on the real line for some $p<\infty,$ where ${\mathfrak S_p},$
$p\in(0, \infty],$ are the standard Schatten -- von Neumann classes of compact operators.

3. An alternative sufficient condition is the condition $\alpha\in\mathfrak S_2$ (and therefore $B_\varkappa \in \mathfrak S_1$), or, more generally, $\alpha M(z)^{-1}\alpha\in\mathfrak S_1,$ see \cite{MR1036844} for details.

4. Following from the analysis above, the existence and completeness of the wave operators for the par $A_\varkappa,$ $A_0$ is closely linked to the condition of $\alpha$ having a ``relative Hilbert-Schmidt property'' with respect to $M(z).$
Recalling that $B_\varkappa=\alpha\varkappa\alpha/2,$ this is not always feasible to expect. Nevertheless, by appropriately modifying the boundary triple, the situation can often be rectified. For example, if  $C_\varkappa=C_0+\alpha\varkappa\alpha/2,$ where $C_0$ and  $\varkappa$ are bounded and $\alpha\in{\mathfrak S}_2,$ replaces the operator $B_\varkappa$ in (\ref{eq:b-kappa-def}), then one ``shifts" the boundary triple ({\it cf.} the proof of Lemma \ref{lem:singular-projection-sa-part}):
$\widehat{\Gamma}_0=\Gamma_0,$ $\widehat{\Gamma}_1=\Gamma_1-C_0\Gamma_0.$ One thus obtains that in the new triple $({\mathcal K}, \widehat{\Gamma}_0, \widehat{\Gamma}_1)$ the operator $A_{\varkappa}$  coincides with the extension  corresponding to the boundary operator $B_\varkappa=\alpha\varkappa\alpha/2,$ whereas the Weyl-Titchmarsh function $M(z)$ undergoes a shift to the function $M(z)-C_0.$ The proof of Theorem 6.1 remains intact, while Part 3 of this remark yields that the condition $\alpha (M(z)-C_0)^{-1}\alpha\in{\mathfrak S}_1$ guarantees the existence and completeness of the wave operators for the pair $A_{C_0},$ $A_{C_\varkappa}.$ The fact that the operator $A_0$ here is replaced by the operator $A_{C_0}$ reflects the standard argument that the complete scattering theory for a pair of operators requires that the operators forming this pair are ``close enough" to each other.

\end{remark}

%%%%%%%%%%%%%%%%%%%%%%%%%%%%%%%%%%%%%%%%%%%%%%%%%%%%%%%%%%

Finally, the scattering operator $\Sigma$ for the pair $A_\varkappa,$
$A_0$
%of two self-adjoint extensions parameterized by matrices $0,B$,
%respectively as
is defined by
$$
\Sigma=W_{+}^{-1}(A_\varkappa,A_0) W_{-}(A_\varkappa,A_0).
$$
The above formulae for the wave operators lead ({\it cf.} \cite{MR573902}) to the following formula for the action of $\Sigma$ in the model representation:
\begin{equation}
\Phi\Sigma\Phi^*P_K \binom{\tilde g}{g}=
P_K \binom
{-(I+\chi_\varkappa^-(S-I))^{-1}(I+\chi_\varkappa^+(S^*-I))g}
{(I+S^*)^{-1}(I+S)(I+\chi_\varkappa^-(S-I))^{-1}(I+\chi_\varkappa^+(S^*-I))g},
\label{last_formula}
\end{equation}
whenever $\Phi^*P_K \binom{\tilde g}{g}\in\widetilde N_{\rm e}^0$. In fact, as explained
above, this representation holds on a dense linear set in
$\widetilde N_{\rm e}^0$ within the conditions of Theorem \ref{thm:existence-completeness-wave-operators}, which guarantees that
all the objects on the right-hand side of the formula (\ref{last_formula}) are correctly defined.
%which is provided that all the objects in the right-hand side of
%the formula can be attributed sense and the wave operators therefore
%exist.

\section{Spectral representation for the absolutely continuous part of the operator $A_0$}
\label{sec:spectral-repr-ac}
The identity
$$
\biggl\|P_K\binom{\tilde g}{g}\biggr\|^2_{\mathfrak H}=\bigl\langle(I-S^*S)\tilde g, \tilde g \bigr\rangle
$$
which is derived in the same way as in \cite[Section 7]{MR573902} for all
$P_K \binom{\tilde g}{g}\in\widetilde N_{\rm e}^0$ and is equivalent to the
condition $(\tilde g+S^*g)+(S\tilde g+g)=0,$ see (\ref{lyubimaya_formula}), allows us to consider the
isometry
$F: \Phi\widetilde N_{\rm e}^0\mapsto L^2(E; I-S^*S)$
defined by the formula
\begin{equation}
FP_K \binom{\tilde g}{g}=\tilde g.
\label{F_def}
\end{equation}
Here $L^2(E; I-S^*S)$ is the Hilbert space of $E$-valued functions on
$\mathbb R$ square summable with the matrix ``weight'' $I-S^*S,$ {\it cf.} (\ref{mathfrakH}). Similarly, the formula
$$
F_*P_K \binom{\tilde g}{g}= g
$$
defines an isometry $F_*$ from $\Phi\widetilde N_{\rm e}^0$ to $L^2(E; I-SS^*).$

\begin{lemma}
Suppose that the assumptions of Theorem \ref{thm:existence-completeness-wave-operators} hold. Then the ranges of the operators $F$ and $F_*$ are dense in the spaces $L^2(E; I-S^*S)$ and $L^2(E; I-SS^*),$ respectively.
\end{lemma}
\begin{proof}
Indeed, for all $\tilde g \in L^2(E; I-S^*S)$ and $g=-S\tilde{g}$ one has $(\tilde g, g)\in\mathcal H$
with $\|(\tilde g,g)\|_{\mathcal{H}}=\|\tilde g\|_{L^2(E; I-S^*S)}.$
By repeating the proof of Theorem \ref{thm:existence-completeness-wave-operators},
%Lemma~\ref{lem:scalar-multiples},  $I+S^*(\bar{\cdot})$ has a
%scalar multiple in $\mathbb C_-,$ which means that
%$(I+S^*(\bar{z}))\Omega(z)=\Omega(z)(I+S^*(\bar{z}))=d(z)$
%for all $z\in \mathbb C_-,$ with a bounded analytic
%operator-valued function $\Omega$ and a bounded scalar analytic
%function $d$. Then
the operator $I+S^*$ is boundedly invertible almost everywhere on ${\mathbb R}.$
%for for almost all $s\in{\mathbb R}.$

Further, consider $\mathbbm{1}_n(k),$ the indicator of the set $\{k\in \mathbb R: \|(I+S^*(k))^{-1}\|\leq n\}.$
%Assuming that the matrix $I+S^*(\bar{\cdot})$ has a scalar multiple in ${\mathbb C}_-$
For $\tilde{g}\in L^2(E; I-S^*S)$ and, as above, $g=-S\tilde{g},$ one has $\mathbbm{1}_n(\tilde g, -(I+S^*)^{-1}(I+S)\tilde g)\in\mathcal H$, since
$$
\mathbbm{1}_n\binom{\tilde g}{-(I+S^*)^{-1}(I+S)\tilde g}-\mathbbm{1}_n\binom{\tilde g}{g}\quad\quad\quad\quad\quad\quad\quad\quad\quad\quad\quad\quad
$$
$$
\quad\quad\quad\quad\quad\quad\quad\quad\quad\quad\quad\quad=\binom{0}{-\mathbbm{1}_n (I+S^*)^{-1}\bigl[(S\tilde g+g)+(\tilde g+S^* g)\bigr]}
\in \binom{0}{L^2(E)}.
$$
%where as above $g=-S\tilde g$ and therefore both $(\tilde g,g)$ and
%$\mathbbm{1}_n(\tilde g,g)$ are in $\mathcal H$.
Finally, the set $\{\mathbbm{1}_n\tilde g\}$ is dense in $L^2(E; I-S^*S)$ by the Lebesgue
dominated convergence theorem, whereas
$P_K\mathbbm{1}_n(\tilde g, -(I+S^*)^{-1}(I+S)\tilde g)\in \widetilde N_{\rm e}^0$ by direct calculation.
\end{proof}

\begin{corollary}
\label{cor_thm}
The operator $F,$ respectively $F_*,$ admits an extension to the
unitary mapping between $\Phi N_{\rm e}^0$ and
%the Hilbert space
$L^2(E; I-S^*S),$ respectively $L^2(E; I-SS^*).$
%In the same way, the operator $\mathbf{F}_*$ can be
%treated as a unitary mapping of $N_{\rm e}$ onto the Hilbert space
%$L^2(E;(1-SS^*))$.
\end{corollary}

It follows that the operator $(A_0-z)^{-1}$ (see notation (\ref{eq:a-kappa-def}))
considered on $\widetilde N_{\rm e}^0$ acts as the multiplication by
$(k-z)^{-1},$ $k\in{\mathbb R},$
% in either of the representations given above, i.e.,
both in $L^2(E; I-S^*S)$ and $L^2(E; I-SS^*)$. In particular,
%t follows,
%that
if one considers the absolutely continuous ``part'' of the
operator $A_0$, namely the operator $A_0^{({\rm e})}:=A_0|_{N_{\rm e}^0},$ then  $F\Phi A_0^{({\rm e})}\Phi^*F^*$ and
%=[k], \quad
$F_*\Phi A_0^{({\rm e})}\Phi^*F_*^*
$
%are
%where $[k]$ denotes the
are the operators of multiplication by the independent variable in the
spaces $L^2(E; I-S^*S)$ and $L^2(E; I-SS^*)$, respectively.

%Th  spectral since one still faces
%operator ``weights'' in the definitions of spaces of square summable
%functions.
In order to obtain a spectral representation from the above result, it
is necessary to diagonalise the weights in the definitions of the
above $L^2$-spaces. The corresponding transformation is straightforward when
$\alpha=\sqrt{2}I.$ (This choice of $\alpha$ satisfies the conditions of Theorem \ref{thm:existence-completeness-wave-operators} {\it e.g.} when the boundary space $\mathcal K$ is finite-dimensional.
%, which is the case we deal with in the application discussed in Sections \ref{sec:quantum-graphs}, \ref{sec:inverse-%scattering}.
The corresponding diagonalisation in the general setting of non-negative, bounded $\alpha$ will be treated elsewhere.)  In this particular case one has
\begin{equation}
S=(M-iI)(M+iI)^{-1},
\label{SviaM}
\end{equation}
and
consequently
\begin{equation}
I-S^*S=-2i (M^*-iI)^{-1}(M-M^*)(M+iI)^{-1}
\label{weightform}
\end{equation}
and
$$
I-SS^*=2i (M+iI)^{-1}(M^*-M)(M^*-iI)^{-1}.
$$
Introducing the unitary transformations
\begin{equation}
\label{unit_trans1}
G: L^2(E; I-S^*S)\mapsto L^2(E; -2i(M-M^*)),
\end{equation}
\begin{equation}
\label{unit_trans2}
G_*: L^2(E; I-SS^*)\mapsto L^2(E; -2i(M-M^*))
\end{equation}
 by the
formulae $g\mapsto (M+iI)^{-1}g$ and
$g\mapsto (M^*-iI)^{-1} g$ respectively, one arrives at the fact that
$
GF\Phi A_0^{({\rm e})}\Phi^* F^*G^*$
and
% \quad
$G_*F_*\Phi A_0^{({\rm e})}\Phi^*F_*^*G_*^*
$
%where $[k]$ denotes
are the operators of multiplication by the independent
variable in the space $L^2(E;-2i(M-M^*))$.

\begin{remark}
The weight $M^*-M$ can be assumed to be naturally diagonal in the setting of quantum graphs, see \cite{CherednichenkoKiselevSilva} ({\it cf.} \cite{MR3484377, MR3430381}), including the situation of an infinite number of semi-infinite edges.
\end{remark}

The above result
%formulated in Corollary \ref{cor_thm}
only pertains to the absolutely
continuous part of the self-adjoint operator $A_0$, unlike {\it e.g.}
the passage to the classical von Neumann direct integral, under which
the whole of the self-adjoint operator gets mapped to the
multiplication operator in a weighted $L^2$-space (see {\it e.g.}
\cite[Chapter 7]{MR1192782}). Nevertheless, it proves useful in scattering theory, since it
yields an explicit expression for the scattering matrix $\widehat{\Sigma}$ for the pair
$A_\varkappa,$ $A_0,$ which is the image of the scattering operator
$\Sigma$ in the spectral representation of the operator $A_0.$ Namely, we prove the following statement.

\begin{theorem}
The following formula holds:
%\footnote{
%}
\begin{equation}
\label{scat2}
\widehat{\Sigma}=GF\Sigma(GF)^*=
(M-\varkappa)^{-1}(M^*-\varkappa)(M^*)^{-1}M,
\end{equation}
where the right-hand side
%s of (\ref{scat1}),
%(\ref{scat2})
represents the operator of multiplication by the corresponding function in the space $L^2(E;-2i(M-M^*))$.
\end{theorem}
\begin{proof}

Using the definition (\ref{F_def}) of the isometry $F$ along with the relationship (\ref{lyubimaya_formula}) between $\widetilde{g}$ and $g$ whenever
$P_K\binom{\widetilde{g}}{g}\in\Phi \widetilde{N}_{\rm e}^\varkappa$ with $\varkappa=0,$
%$\binom{\widetilde g}{g}
we obtain from (\ref{last_formula}):
\begin{equation}
\label{scat1}
F\Sigma F^*=
\bigl(I+\chi_\varkappa^-(S-I)\bigr)^{-1}\bigl(I+\chi_\varkappa^+(S^*-I)\bigr)(I+S^*)^{-1}(I+S),
\end{equation}
where the right-hand side
%s of (\ref{scat1}),
%(\ref{scat2})
represents the operator of multiplication by the corresponding function.
%and hence, after a simplification,

Furthermore, substituting the expression (\ref{S_definition}) for $S$ in terms of $M$ implies that $F\Sigma F^*$ is the operator of multiplication by
\[
(M+iI)(M-\varkappa)^{-1}(M^*-\varkappa)(M^*)^{-1}M(M+iI)
\]
in the space $L^2({\mathcal K}; I-S^*S).$ Using (\ref{weightform}), we now obtain the following identity for all $f, g\in L^2({\mathcal K}; I-S^*S):$
\[
\langle F\Sigma F^*f,g\rangle_{L^2({\mathcal K}; I-S^*S)}=\bigl\langle (I-S^*S)(M+iI)(M-\varkappa)^{-1}(M^*-\varkappa)(M^*)^{-1}M(M+iI)f, g\bigr\rangle
\]
\[
=\bigl\langle -2i (M^*-iI)^{-1}(M-M^*)(M+iI)^{-1}(M+iI)(M-\varkappa)^{-1}(M^*-\varkappa)(M^*)^{-1}M(M+iI)f, g\bigr\rangle
\]
\[
=\bigl\langle -2i(M-M^*)(M-\varkappa)^{-1}(M^*-\varkappa)(M^*)^{-1}M(M+iI)f, (M+iI)g\bigr\rangle,
\]
which is equivalent to (\ref{scat2}), in view of the definition of the operator $G.$
\end{proof}

%We remark, that this is the scattering matrix of $A_B$ relative to
%$A_0$ in the spectral representation of the latter which is the space
%of ${\mathcal K}$-valued square summable functions on the real line with the
%operator-valued ``weight'' $-2i(M-M^*)$.

%%In applications to quantum graphs it may turn out that the
%%operator weight $-2i(M-M^*)$ (see (\ref{unit_trans1}), (\ref{unit_trans2})) is degenerate:
%. In particular, in the setting
%of quantum graphs it comes handy to express
%%more precisely, $M(s)-M(s)^*=2i\sqrt{s}P_{\rm e},$ $s\in{\mathbb R},$ where
%%$P_{\rm e}$ is the orthogonal projection onto the subspace of ${\mathcal K}$ corresponding
%%to the set of ``external'' vertices of the graph, {\it i.e.} those vertices
%%to which semi-infinite edges are attached. Next, we describe the notation pertaining to the quantum graph setting.

\section*{Acknowledgements}

KDC is grateful for the financial support of
%has been supported by
%the Leverhulme Trust (Grant RPG--167 ``Dissipative and non-self-adjoint problems'') and
the Engineering and Physical Sciences Research Council: Grant EP/L018802/2 ``Mathematical foundations of metamaterials: homogenisation, dissipation and operator theory''. AVK has been partially supported by the RFBR grant 16-01-00443-a and the Russian Federation Government megagrant 14.Y26.31.0013. LOS has been partially supported by UNAM-DGAPA-PAPIIT IN105414 and SEP-CONACYT 254062.

We express deep gratitude to Professor Sergey Naboko for his reading the manuscript and making a number of valuable remarks.

We also thank the referee for a number of useful suggestions, which have helped us improve the manuscript.

\appendix
\noindent\textbf{Proof of Theorem~\ref{thm:rhyzhov2.5}.}

We prove Theorem~\ref{thm:rhyzhov2.5}(i). The proof of
Theorem~\ref{thm:rhyzhov2.5}(ii) is carried out along the same lines.

For any $(v_-,u,v_+)$ in the space $\mathscr{H}$ given in
(\ref{eq:wider-hilbert-pavlov}), consider the mappings
$\mathscr{F}_\pm:\mathscr{H}\to L^2(\reals,E)$
introduced in \cite[Sec.\,2.1]{MR2330831} following the corresponding
definitions in \cite{MR573902} and given by
\begin{align}
  \mathscr{F}_+(v_-,u,v_+)&=-\frac{1}{\sqrt{2\pi}}
\lim_{\epsilon\searrow 0}\alpha\Gamma_0
  (A_{\I I}-(\cdot-\I\epsilon)I)^{-1}u+S^*\hat{v}_-+\hat{v}_+
  \label{eq:f+def}\\
  \mathscr{F}_-(v_-,u,v_+)&=-\frac{1}{\sqrt{2\pi}}
 \lim_{\epsilon\searrow 0} \alpha\Gamma_0
  (A_{\I I}^*-(\cdot+\I\epsilon)I)^{-1}u+\hat{v}_-+S\hat{v}_+\,,
\label{eq:f-def}
\end{align}
where $\hat{v}_\pm$ are the Fourier transforms of $v_\pm\in
L^2(\reals_\pm,E)$ extended by zero to $L^2(\reals,E)$. Note that the
limits exist almost everywhere due to (\ref{eq:naboko-thm-1}).

According to \cite[Thm.\,2.3]{MR2330831}, if
$\binom{\widetilde{g}}{g}=\Phi h$, then
\begin{equation}
 \label{eq:action-f}
  \mathscr{F}_+ h=\widetilde{g}+S^*g\,,\qquad
  \mathscr{F}_- h=S\widetilde{g}+g\,.
\end{equation}
Therefore, for proving Theorem~\ref{thm:rhyzhov2.5}(i), one should
establish the validity of the identities:
\begin{equation}
 \label{eq:validating-equalities}
  \mathscr{F}_\pm(A_\varkappa-z I)^{-1}\Phi^{-1}\binom{\widetilde{g}}{g}
=\mathscr{F}_\pm\Phi^{-1}P_K\frac{1}{\cdot-z}
\binom{\widetilde{g}}
{g-\chi_\varkappa^+\Theta^{-1}_\varkappa(z)(\widetilde{g}+S^*g)(z)}
\end{equation}
for $z\in\complex_-\cap\rho(A_\varkappa)$. First we compute the left-hand-side of
\eqref{eq:validating-equalities}.
It follows from Lemma~\ref{lem:ryzhov-formulae}(\ref{eq:ii-kappa}),
(\ref{eq:kappa-ii}) that, for
$z,\lambda\in\complex_-\cap\rho(A_\varkappa)$ and $h\in\mathcal{H}$,
\begin{align*}
  \alpha\Gamma_0(A_{\I I}&-zI)^{-1}(A_\varkappa-\lambda I)^{-1}h\\
&=\Theta_\varkappa(z)\alpha\Gamma_0(A_\varkappa -z
I)^{-1}(A_\varkappa-\lambda I)^{-1}h\\
&=\frac{1}{z-\lambda}\Theta_\varkappa(z)\alpha\Gamma_0
\left[(A_\varkappa-z I)^{-1}-(A_\varkappa-\lambda I)^{-1}\right]h\\
&=\frac{1}{z-\lambda}\left[\alpha\Gamma_0(A_{\I I}-zI)^{-1}-
\Theta_\varkappa(z)\alpha\Gamma_0(A_\varkappa-\lambda
I)^{-1}\right]h\\
&=\frac{1}{z-\lambda}\left[\alpha\Gamma_0(A_{\I I}-zI)^{-1}-
\Theta_\varkappa(z)\Theta_\varkappa ^{-1}(\lambda)\alpha\Gamma_0(A_{\I I}-\lambda I)^{-1}\right]h\,.
\end{align*}
Let $z=k-\I\epsilon$ with $k\in\reals$, then it follows from the
computation above that
\begin{align*}
  &\lim_{\epsilon\searrow 0}\alpha\Gamma_0(A_{\I I}-(k-\I\epsilon)I)^{-1}
(A_\varkappa-\lambda I)^{-1}h\\ &=\lim_{\epsilon\searrow 0}
\frac{1}{(k-\I\epsilon)-\lambda}\left[\alpha\Gamma_0(A_{\I I}-(k-\I\epsilon)I)^{-1}-
\Theta_\varkappa(k-\I\epsilon)\Theta_\varkappa ^{-1}(\lambda)\alpha\Gamma_0(A_{\I I}-\lambda I)^{-1}\right]h\,.
\end{align*}
Substituting (\ref{eq:f+def}) into the last equality, one has
\begin{equation*}
  \mathscr{F}_+(A_\varkappa-\lambda I)^{-1}h=
 \frac{1}{\cdot-\lambda}\left[\mathscr{F}_+h-
\Theta_\varkappa(\cdot)\Theta_\varkappa ^{-1}(\lambda)\mathscr{F}_+h(\lambda)
\right]\,.
\end{equation*}
Hence, in view of \eqref{eq:action-f}, one concludes
\begin{equation}
\label{eq:f+on-resolvent}
  \mathscr{F}_+(A_\varkappa-\lambda I)^{-1}\Phi^{-1}\binom{\widetilde{g}}{g}=
 \frac{1}{\cdot-\lambda}\left[
 \widetilde{g}+S^*g-\Theta_\varkappa(\cdot)\Theta_\varkappa
   ^{-1}(\lambda)(\widetilde{g}+S^*g)(\lambda)\right]\,.
\end{equation}

On the basis of Lemma~\ref{lem:ryzhov-formulae}(\ref{eq:ii-kappa+}),
(\ref{eq:kappa-ii+}) and reasoning in the same fashion as was done to
obtain (\ref{eq:f+on-resolvent}), one verifies
\begin{equation}
  \label{eq:f-on-resolvent}
  \mathscr{F}_-(A_\varkappa-\lambda I)^{-1}\Phi^{-1}\binom{\widetilde{g}}{g}=
 \frac{1}{\cdot-\lambda}\left[
 S\widetilde{g}+g-\widehat{\Theta}_\varkappa(\cdot)\Theta_\varkappa
   ^{-1}(\lambda)(\widetilde{g}+S^*g)(\lambda)\right]\,.
\end{equation}

Let us focus on the right hand side of
\eqref{eq:validating-equalities}. Note that
\begin{align}
  &P_K\frac{1}{\cdot-z}
\binom{\widetilde{g}}
{g-\chi_\varkappa^+\Theta^{-1}_\varkappa(z)(\widetilde{g}+S^*g)(z)}\nonumber\\[3mm]
 &=\binom{\frac{\widetilde{g}}{\cdot-z}-P_+\frac{1}{\cdot-z}
  [\widetilde{g}+S^*g-S^*\chi_\varkappa^+\Theta^{-1}_\varkappa(z)(\widetilde{g}+S^*g)(z)]}{\frac{1}{\cdot-z}(g-\chi_\varkappa^+\Theta^{-1}_\varkappa(z)(\widetilde{g}+S^*g)(z))-P_-\frac{1}{\cdot-z}[S\widetilde{g}+g-\chi_\varkappa^+\Theta^{-1}_\varkappa(z)(\widetilde{g}+S^*g)(z)]}\nonumber\\[3mm]
&=\frac{1}{\cdot-z}\binom{\widetilde{g}-(\widetilde{g}+S^*g)(z) +
S^*(\cc{z})\chi_\varkappa^+\Theta^{-1}_\varkappa(z)(\widetilde{g}+S^*g)(z)}
{g-\chi_\varkappa^+\Theta^{-1}_\varkappa(z)(\widetilde{g}+S^*g)(z)}
\label{eq:rhs}
\end{align}
where (\ref{eq:pk-action}) is used in the first equality and in the
second the fact that if $f$ is a function in ${H}_-^2$, then, for any $z\in\complex_-$,
\begin{equation}
\label{eq:focus-pocus}
  P_+\left(\frac{f}{\cdot-z}\right)=P_+\left(\frac{f+f(z)-f(z)}{\cdot-z}\right)=P_+\left(\frac{f(z)}{\cdot-z}\right)=\frac{f(z)}{\cdot-z}\,.
\end{equation}
Now, apply $\mathscr{F}_+\Phi^{-1}$ to
(\ref{eq:rhs}) taking into account (\ref{eq:action-f}):
\begin{align*}
  &\mathscr{F}_+\Phi^{-1}\frac{1}{\cdot-z}\binom{\widetilde{g}-(\widetilde{g}+S^*g)(z) +
S^*(\overline{z})\chi_\varkappa^+\Theta^{-1}_\varkappa(z)(\widetilde{g}+S^*g)(z)}
{g-\chi_\varkappa^+\Theta^{-1}_\varkappa(z)(\widetilde{g}+S^*g)(z)}\\
&=\frac{1}{\cdot-z}[\widetilde{g}+S^*g-(\widetilde{g}+S^*g)(z)+
(S^*(\overline{z})-S^*)\chi_\varkappa^+\Theta^{-1}_\varkappa(z)(\widetilde{g}+S^*g)(z)]\\
&=\frac{1}{\cdot-z}[\widetilde{g}+S^*g-
(\Theta_\varkappa(z)-(S^*(\overline{z})-S^*)\chi_\varkappa^+)\Theta^{-1}_\varkappa(z)(\widetilde{g}+S^*g)(z)]\\
&=\frac{1}{\cdot-z}[\widetilde{g}+S^*g-\Theta(\cdot)\Theta^{-1}_\varkappa(z)(\widetilde{g}+S^*g)(z)].
\end{align*}
By combining the last equality with (\ref{eq:f+on-resolvent}), we have
established the first identity in (\ref{eq:validating-equalities}).

Now, if one applies $\mathscr{F}_-\Phi^{-1}$ to
(\ref{eq:rhs}), then, in view of (\ref{eq:action-f}), one has
\begin{align*}
  &\mathscr{F}_-\Phi^{-1}\frac{1}{\cdot-z}\binom{\widetilde{g}-(\widetilde{g}+S^*g)(z) +
S^*(\overline{z})\chi_\varkappa^+\Theta^{-1}_\varkappa(z)(\widetilde{g}+S^*g)(z)}
{g-\chi_\varkappa^+\Theta^{-1}_\varkappa(z)(\widetilde{g}+S^*g)(z)}\\
&=\frac{1}{\cdot-z}[S\widetilde{g}+g-S(\widetilde{g}+S^*g)(z)-
(I-SS^*(\overline{z}))\chi_\varkappa^+\Theta^{-1}_\varkappa(z)(\widetilde{g}+S^*g)(z)]\\
&=\frac{1}{\cdot-z}[S\widetilde{g}+g-
(S\Theta_\varkappa(z)+\chi_\varkappa^+-SS^*(\overline{z})\chi_\varkappa^+)\Theta^{-1}_\varkappa(z)(\widetilde{g}+S^*g)(z)]\\
&=\frac{1}{\cdot-z}[S\widetilde{g}+g-
(S\chi_\varkappa^-+\chi_\varkappa^-)\Theta^{-1}_\varkappa(z)(\widetilde{g}+S^*g)(z)]\\
&=\frac{1}{\cdot-z}[S\widetilde{g}+g-\widehat{\Theta}_\varkappa(\cdot)\Theta^{-1}_\varkappa(z)(\widetilde{g}+S^*g)(z)]
\end{align*}
Thus, after comparing this last equality with
(\ref{eq:f-on-resolvent}), we arrive at the second identity in
(\ref{eq:validating-equalities}).

\

\noindent\textbf{Proof of Theorem~\ref{lem:similar-to-naboko-thm-4}.}

Let us first show that the following inclusion holds
\begin{equation*}
 \mathfrak{N}^\varkappa_\pm\subset\left\{\binom{\widetilde{g}}{g}\in\mathfrak{H}:
  \Phi(A_{\varkappa}-z I)^{-1}\Phi^*P_K\binom{\widetilde{g}}{g}=
P_K\frac{1}{\cdot-z}\binom{\widetilde{g}}{g}\text{ for all } z\in\complex_\pm\right\}
\end{equation*}
Consider $z\in\complex_-\cap\rho(A_\varkappa)$.
  By (\ref{eq:pk-action}) and Theorem~\ref{thm:rhyzhov2.5}, one has
  \begin{align*}
  &\Phi(A_\varkappa-z I)^{-1}\Phi^{-1}P_K\binom{\widetilde{g}}{g}
=\Phi(A_\varkappa-z
I)^{-1}\Phi^{-1}\binom{\widetilde{g}-P_+(\widetilde{g}+S^*g)}{g-P_-(S\,\widetilde{g}+g)}\\[5mm]
&\!\!=P_K\frac{1}{\cdot-z}\binom{\widetilde{g}-P_+(\widetilde{g}+S^*g)}
{g-P_-(S\widetilde{g}+g)-\chi_\varkappa^+\Theta_\varkappa^{-1}(z)\left[\widetilde{g}-
P_+(\widetilde{g}+S^*g)+S^*(g-P_-(S\widetilde{g}+g))\right](z)}
 \end{align*}
 where
 \begin{equation*}
   \left[\widetilde{g}-
P_+(\widetilde{g}+S^*g)+S^*(g-P_-(S\widetilde{g}+g))\right](z)
 \end{equation*}
is to be understood as the analytic continuation into the lower
half-plane of the function
\begin{equation}
  \label{eq:function_in_h-}
  \widetilde{g}-
P_+(\widetilde{g}+S^*g)+S^*(g-P_-(S\widetilde{g}+g))=P_-(\widetilde{g}+S^*g)-S^*P_-(S\widetilde{g}+g),
\end{equation}
which is clearly an element of ${H}^2_-(E).$
%%The fact that \eqref{eq:function_in_h-} holds follows from (\ref{characterise_K}) and \eqref{eq:pk-action}.
%%Now, one rewrites the expression for this function using the fact that
%%$I_{L^2(E)}-P_-=P_+$ (i.\,e., ${H}^2_+(E)$ is the orthogonal complement of
%%${H}^2_-(E)$ in $L^2(\reals,E)$):
%% \begin{align*}
%%   \widetilde{g}-
%%P_+(\widetilde{g}+S^*g)+S^*(g-P_-(S\widetilde{g}+g))&=
%%  (I_{L^2(E)}- P_+)(\widetilde{g}+S^*g)-S^*P_-(S\widetilde{g}+g)\\
%% &=P_-(\widetilde{g}+S^*g)-S^*P_-(S\widetilde{g}+g)\,.
%% \end{align*}
%% Note that this equality makes evident \eqref{eq:function_in_h-}.
Thus,
\begin{equation}
\label{eq:initial-expression}
%\begin{split}
  %&
  \Phi(A_\varkappa-z I)^{-1}\Phi^{-1}P_K\binom{\widetilde{g}}{g}
%\\
%&
= P_K\frac{1}{\cdot-z}\binom{\widetilde{g}-P_+(\widetilde{g}+S^*g)}
{g-P_-(S\widetilde{g}+g)-
%\chi_\varkappa^+\Theta_\varkappa^{-1}(z)
%\left[P_-(\widetilde{g}+S^*g)-S^*P_-(S\widetilde{g}+g)\right](z)}.
\gamma(z)}
%\\
%&=P_K\frac{1}{\cdot-z}\binom{\widetilde{g}-P_+(\widetilde{g}+S^*g)}
%{g-P_-(S\widetilde{g}+g)-\gamma(z) + P_-(S\widetilde{g} +
%  g)(z) - P_-(S\widetilde{g} + g)(z)}
%\end{split}
\end{equation}
where
\begin{equation}
 \label{eq:kappa-def}
\gamma(z):=\chi_\varkappa^+\Theta_\varkappa^{-1}(z)\bigl(P_-(\widetilde{g}+S^*g)(z)-S^*P_-(S\widetilde{g}+g)(z)\bigr).
\end{equation}
%and a zero has been added in the second equality of
%\eqref{eq:initial-expression}.
The following lemma is needed to simplify the form of $\gamma(z)$.

\begin{lemma}
\label{final_lemma}
For all $\binom{\widetilde{g}}{g}\in\mathfrak{H}$ the following identity holds:
\[
%\chi_\varkappa^+\Theta_\varkappa^{-1}(z)
%\left[P_-(\widetilde{g}+S^*g)-S^*P_-(S\widetilde{g}+g)\right](z)
\gamma(z)=-P_-(S\widetilde{g}+g)(z)\ \ \ \ \ \forall z\in{\mathbb C_-}.
\]
\end{lemma}
\begin{proof}
\[
\chi_\varkappa^+\Theta_\varkappa^{-1}(z)\bigl(P_-(\widetilde{g}+S^*g)(z)-S^*(\cc{z})P_-(S\widetilde{g}+g)(z)\bigr)
\]
%\end{equation*}
%\]
%Combining this with the formula (\ref{eq:theta}) we obtain
%\[
% \chi_\varkappa^+\alpha\Gamma_0(A_{\varkappa}-zI)^{-1}u
% \]
% \[
% =-\sqrt{2\pi}\chi_+^\varkappa\Theta_\varkappa^{-1}(z)\biggl(P_-(\widetilde{g}+S^*g)(z)-S^*(\overline{z})P_-(S\widetilde{g}+g)(z)\biggr)
%\]
\[
=\chi_\varkappa^+\bigl(I+i\alpha(B_\varkappa-M(z))^{-1}\alpha\chi_\varkappa^+\bigr)\bigl(P_-(\widetilde{g}+S^*g)(z)-S^*(\overline{z})P_-(S\widetilde{g}+g)(z)\bigr)
\]
\[
=\bigl(I+i\chi_\varkappa^+\alpha(B_\varkappa-M(z))^{-1}\alpha\bigr)\chi_\varkappa^+\bigl(P_-(\widetilde{g}+S^*g)(z)-S^*(\overline{z})P_-(S\widetilde{g}+g)(z)\bigr)
\]
\[
=\bigl(I+\chi_\varkappa^+(S^*(\overline{z})-I)\bigr)^{-1}\bigl(\chi_\varkappa^+P_-(\widetilde{g}+S^*g)(z)-\chi_\varkappa^+S^*(\overline{z})P_-(S\widetilde{g}+g)(z)\bigr)
\]
\[
=\bigl(I+\chi_\varkappa^+(S^*(\overline{z})-I)\bigr)^{-1}\bigl(-\chi_\varkappa^-P_-(S\widetilde{g}+g)(z)-\chi_\varkappa^+S^*(\overline{z})P_-(S\widetilde{g}+g)(z)\bigr)
\]
\[
=\bigl(I+\chi_\varkappa^+(S^*(\overline{z})-I)\bigr)^{-1}\bigl(-\chi_\varkappa^--\chi_\varkappa^+S^*(\overline{z})\bigr)P_-(S\widetilde{g}+g)(z)=-P_-(S\widetilde{g}+g)(z),
\]
where we use the fact that
\[
I+i\chi_\varkappa^+\alpha(B_\varkappa-M(z))^{-1}\alpha=\bigl(I+\chi_\varkappa^+(S^*(\overline{z})-I)\bigr)^{-1},
\]
proved in a similar way to (\ref{eq:theta-inverse}).
\end{proof}

%Taking into account (\ref{eq:alternate-for-theta}) and
%(\ref{eq:alternate-for-theta-hat}), one rewrites the following
%expression
%\begin{align*}
%  &P_-(S\widetilde{g} + g)(z) + \gamma(z)\\
%&=\chi_\varkappa^+\Theta_\varkappa^{-1}(z)\left[\Theta_\varkappa(z)(\chi_\varkappa^+)^{-1}P_-(S\widetilde{g}+g)(z)+
%P_-(\widetilde{g}+S^*g)(z)-S^*P_-(S\widetilde{g}+g)(z)\right]\\
%&=\chi_\varkappa^+\Theta_\varkappa^{-1}(z)\left[(\chi_\varkappa^-+S^*(\cc{z})\chi_\varkappa^+)
%(\chi_\varkappa^+)^{-1}P_-(S\widetilde{g}+g)(z)+
%P_-(\widetilde{g}+S^*g)(z)-S^*P_-(S\widetilde{g}+g)(z)\right]\\
%&=\chi_\varkappa^+\Theta_\varkappa^{-1}(z)\left[\chi_\varkappa^-(\chi_\varkappa^+)^{-1}P_-(S\widetilde{g}+g)(z)+
%P_-(\widetilde{g}+S^*g)(z)\right]\\
%&=\chi_\varkappa^+\Theta_\varkappa^{-1}(z)(\chi_\varkappa^+)^{-1}\left[
%\chi_\varkappa^-P_-(S\widetilde{g}+g)(z)+
%\chi_\varkappa^+ P_-(\widetilde{g}+S^*g)(z)\right]\,,
%\end{align*}
%where it has been used that $S^*P_-(S\widetilde{g}+g)(z)=
%S^*(\cc{z})P_-(S\widetilde{g}+g)(z)$ and $\chi_\varkappa^-$ commutes
%with $(\chi_\varkappa^+)^{-1}$.

Therefore, for $ \binom{\widetilde{g}}{g}\in\mathfrak{N}_-^\varkappa$
%, then
%\begin{equation*}
%  P_-(S\widetilde{g} + g)(z) + \gamma(z)=0
%\end{equation*}
%for any $z\in\complex_-$.
%Thus,
the expression \eqref{eq:initial-expression} can be re-written as
\begin{align*}
  \Phi(A_\varkappa-z I)^{-1}\Phi^{-1}P_K\binom{\widetilde{g}}{g}&=
P_K\frac{1}{\cdot-z}\binom{\widetilde{g}-P_+(\widetilde{g}+S^*g)}
{g-P_-(S\widetilde{g}+g)+ P_-(S\widetilde{g} + g)(z)}\\[0.5em]
&=P_K\frac{1}{\cdot-z}\left[
\binom{\widetilde{g}}{g}-
\binom{P_+(\widetilde{g}+S^*g)}
{P_-(S\widetilde{g}+g)- P_-(S\widetilde{g} + g)(z)}
\right]
\end{align*}
One completes the proof by observing that
\begin{equation*}
  \frac{P_+(\widetilde{g}+S^*g)}{\cdot-z}\in H^2_+(E),\quad\quad\quad
%\end{equation*}
%and
%\begin{equation*}
  \frac{P_-(S\widetilde{g}+g)-P_-(S\widetilde{g} + g)(z)}{\cdot-z}\in H^2_-(E).
\end{equation*}
We have thus shown that
\begin{equation*}
 \mathfrak{N}^\varkappa_-\subset\left\{\binom{\widetilde{g}}{g}\in\mathfrak{H}:
  \Phi(A_{\varkappa}-z I)^{-1}\Phi^*P_K\binom{\widetilde{g}}{g}=
P_K\frac{1}{\cdot-z}\binom{\widetilde{g}}{g}\text{ for all } z\in\complex_-\right\}\,.
\end{equation*}
The inclusion
\begin{equation*}
 \mathfrak{N}^\varkappa_+\subset\left\{\binom{\widetilde{g}}{g}\in\mathfrak{H}:
  \Phi(A_{\varkappa}-z I)^{-1}\Phi^*P_K\binom{\widetilde{g}}{g}=
P_K\frac{1}{\cdot-z}\binom{\widetilde{g}}{g}\text{ for all } z\in\complex_+\right\}
\end{equation*}
is proved analogously.

To prove the converse inclusion, i.e.
  \begin{equation*}
  \left\{\binom{\widetilde{g}}{g}\in\mathfrak{H}:
  \Phi(A_{\varkappa}-z I)^{-1}\Phi^*P_K\binom{\widetilde{g}}{g}=
P_K\frac{1}{\cdot-z}\binom{\widetilde{g}}{g}\text{ for all } z\in\complex_\pm\right\}
\subset\mathfrak{N}^\varkappa_\pm
\end{equation*}
  one again follows the arguments of \cite[Thm.\,4]{MR573902}.
  According to (\ref{eq:initial-expression}), for all
  $z\in\complex_-\cap\rho(A_\varkappa)$, one has
  \begin{equation*}
    \Phi(A_\varkappa-z I)^{-1}\Phi^{-1}P_K\binom{\widetilde{g}}{g}
= P_K\frac{1}{\cdot-z}\binom{\widetilde{g}-P_+(\widetilde{g}+S^*g)}
{g-P_-(S\widetilde{g}+g)-\gamma(z)},
  \end{equation*}
where $\gamma(z)$ is defined in (\ref{eq:kappa-def}). Denoting $\widehat{\gamma}:=\gamma+P_-(S\widetilde{g}+g),$ it follows from (\ref{eq:pk-action}) that
\begin{align*}
  \Phi(A_\varkappa-z I)^{-1}\Phi^{-1}P_K\binom{\widetilde{g}}{g}
- P_K\frac{1}{\cdot-z}\binom{\widetilde{g}}{g}&= P_K\binom{0}{-\widehat{\gamma}(z)(\cdot-z)^{-1}}\\[0.5em]
&=\binom{P_+(S^*\widehat{\gamma}(z)(\cdot-z)^{-1})}{-\widehat{\gamma}(z)(\cdot-z)^{-1}
+ P_-(\widehat{\gamma}(z)(\cdot-z)^{-1})}
\end{align*}
Furthermore, in view of (\ref{eq:focus-pocus}), one has
\begin{equation*}
  P_+\left[\frac{S^*\widehat{\gamma}(z)}{\cdot-z}\right] = \frac{S^*(\cc{z})\widehat{\gamma}(z)}{\cdot-z}
\end{equation*}
and, clearly,
\begin{equation*}
  P_-\left[\frac{\widehat{\gamma}(z)}{\cdot-z}\right]=0\,.
\end{equation*}
Therefore
\begin{equation*}
  \Phi(A_\varkappa-z I)^{-1}\Phi^{-1}P_K\binom{\widetilde{g}}{g}
- P_K\frac{1}{\cdot-z}\binom{\widetilde{g}}{g}=
\binom{S^*(\cc{z})\widehat{\gamma}(z)(\cdot-z)^{-1}}
{-\widehat{\gamma}(z)(\cdot-z)^{-1}}\,.
\end{equation*}
Since
\begin{equation*}
  \Phi(A_{\varkappa}-z I)^{-1}\Phi^*P_K\binom{\widetilde{g}}{g}=
P_K\frac{1}{\cdot-z}\binom{\widetilde{g}}{g}\text{ for all } z\in\complex_-\,,
\end{equation*}
one has
\begin{equation*}
  \binom{S^*(\cc{z})\widehat{\gamma}(z)(\cdot-z)^{-1}}
{-\widehat{\gamma}(z)(\cdot-z)^{-1}}=0
\end{equation*}
which in its turn implies
\begin{equation*}
\bigl(S^*-S^*(\cc{z})\bigr)\widehat{\gamma}(z)(\cdot-z)^{-1}=0\,.
\end{equation*}
From this equality, by virtue of the assumption that $\ker(\alpha)=\{0\},$ 
%fact that the operator $A_{iI}$ is completely non-self-adjoint, 
one obtains that $\gamma(z)=0$ for
all $z\in\complex_-\cap\rho(A_\varkappa)$ (see details in the proof of \cite[Lem.\,4]{MR0500225}). 
Taking into account the definition of $\widehat{\gamma},$
%(\ref{eq:kappa-def}), 
one arrives at
\begin{equation*}
  \chi_\varkappa^-P_\pm(S\,\widetilde{g} +
      g)+\chi_\varkappa^+ P_\pm(\widetilde{g} + S^*g)=0\,.
\end{equation*}
The inclusion
  \begin{equation*}
  \left\{\binom{\widetilde{g}}{g}\in\mathfrak{H}:
  \Phi(A_{\varkappa}-z I)^{-1}\Phi^*P_K\binom{\widetilde{g}}{g}=
P_K\frac{1}{\cdot-z}\binom{\widetilde{g}}{g}\text{ for all } z\in\complex_+\right\}
\subset\mathfrak{N}^\varkappa_+
\end{equation*}
is proved in a similar way.
\\[.5cm]
\noindent\textbf{Proof of Theorem~\ref{lem:on-smooth-vectors-other-form}.}

To prove the inclusion
\begin{equation*}
\widetilde{N}_-^\varkappa\subset\bigl\{u\in\cH: \chi_\varkappa^+\alpha\Gamma_0(A_{\varkappa}-z
I)^{-1}u\in H^2_-(E)\bigr\}\,,
\end{equation*}
one has to show that $u\in\Phi^*P_K\mathfrak{N}^\varkappa_-$ implies
 $\chi_\varkappa^+\alpha\Gamma_0(A_{\varkappa}-z
I)^{-1}u\in H^2_-(E)$. By \eqref{eq:pk-action}, if $u=\Phi^*P_K\binom{\widetilde{g}}{g}$, then
\begin{equation*}
  \Phi u=\begin{pmatrix}
  \widetilde{g}-P_+(\widetilde{g}+S^*g)\\
  g-P_-(S\,\widetilde{g}+g)
\end{pmatrix}\,.
\end{equation*}
Thus, in view of the inclusion $\binom{\widetilde{g}}{g}\in K,$ it follows from \eqref{eq:action-f} that
\begin{align*}
   \mathscr{F}_+u&=\widetilde{g}-P_+(\widetilde{g}+S^*g)+S^*g-S^*P_-(S\widetilde{g}+g)\\
                &=(I-P_+)(\widetilde{g}+S^*g)-S^*P_-(S\widetilde{g}+g)\\
                &=P_-(\widetilde{g}+S^*g)-S^*P_-(S\widetilde{g}+g)\,.
\end{align*}
By analytic continuation of
$ \mathscr{F}_+u$  into the lower half-plane, taking into account (\ref{eq:f+def}), one arrives at
\begin{equation*}
  \alpha\Gamma_0(A_{\I I}-z I)^{-1}u=-\sqrt{2\pi}\bigl(P_-(\widetilde{g}+S^*g)(z) -
S^*(\cc{z})P_-(S\widetilde{g}+g)(z)\bigr)\quad\quad\forall z\in\complex_-.
\end{equation*}
Combining this with Lemma
\ref{lem:ryzhov-formulae}(\ref{eq:kappa-ii}), we write
\begin{equation*}
  \alpha\Gamma_0(A_{\varkappa}-zI)^{-1}u
=-\sqrt{2\pi}\Theta_\varkappa^{-1}(z)\bigl(P_-(\widetilde{g}+S^*g)(z) -
S^*(\cc{z})P_-(S\widetilde{g}+g)(z)\bigr).
\end{equation*}
%Combining this with the formula (\ref{eq:theta}) we obtain
Finally, using Lemma \ref{final_lemma} from the proof of Theorem \ref{lem:similar-to-naboko-thm-4} above, we obtain
\[
 \chi_\varkappa^+\alpha\Gamma_0(A_{\varkappa}-zI)^{-1}u=\sqrt{2\pi}P_-(S\widetilde{g}+g)(z),
\]

To demonstrate the converse inclusion
\begin{equation*}
\bigl\{u\in\cH: \chi_\varkappa^+\alpha\Gamma_0(A_{\varkappa}-z
I)^{-1}u\in H^2_-(E)\bigr\}\subset\widetilde{N}_-^\varkappa,
\end{equation*}
we show that, whenever $\chi_\varkappa^+\alpha\Gamma_0(A_{\varkappa}-zI)^{-1}u\in H^2_-(E),$ the vector
\[
\binom{\widetilde{g}}{g}=\Phi u-\frac{1}{2\pi}\binom{0}{\alpha\Gamma_0(A_{\varkappa}-zI)^{-1}u}
\]
satisfies
%\mathfrak{N}^\varkappa_\pm:=\left\{\binom{\widetilde{g}}{g}\in\mathfrak{H}:
\[
P_-\left(\chi_\varkappa^+(\widetilde{g}+S^*g)+\chi_\varkappa^-(S\widetilde{g}+g)\right)=0,
\]
and hence $u=\Phi^*P_K\binom{\widetilde{g}}{g}\in \Phi^*P_K\mathfrak{N}^\varkappa_-=\widetilde{N}^\varkappa_{\rm e}.$
Indeed, introducing the notation
\[
\Phi u=:\binom{\widetilde{g_0}}{g_0},\quad\quad\quad\quad h^-:=\frac{1}{2\pi}\alpha\Gamma_0(A_{iI}-zI)^{-1}u,
\]
we have
\begin{equation}
P_-\left(\chi_\varkappa^+(\widetilde{g_0}+S^*(g_0+h^-))+\chi_\varkappa^-(S\widetilde{g_0}+g_0+h^-)\right)
\label{boundary_values}
\end{equation}
\[
=\chi_\varkappa^+(\widetilde{g_0}+S^*g_0)-P_+\chi_\varkappa^+(\widetilde{g_0}+S^*g_0)+P_-\chi_\varkappa^-(S\widetilde{g_0}+g_0)+
\bigl(I+\chi_\varkappa^+(S^*-I)\bigr)h^-
\]
\[
=\chi_\varkappa^+ \mathscr{F}_+u+\bigl(I+\chi_\varkappa^+(S^*-I)\bigr)h^-,
\]
By the analytic continuation into the lower half-plane and using Lemma \ref{lem:ryzhov-formulae}(i), it follows that (\ref{boundary_values}) represents the boundary value
on the real line of the function
\[
-\frac{1}{2\pi}\chi_\varkappa^+\alpha\Gamma_0(A_{iI}-zI)^{-1}u+\bigl(I+\chi_\varkappa^+(S^*(\overline{z})-I)\bigr)h^-(z)
\]
\begin{equation}
=-\frac{1}{2\pi}\chi_\varkappa^+\Theta_\varkappa(z)\alpha\Gamma_0(A_{\varkappa}-zI)^{-1}u+\bigl(I+\chi_\varkappa^+(S^*(\overline{z})-I)\bigr)h^-(z)
\label{penultimate}
\end{equation}
\begin{equation}
=\bigl(I+\chi_\varkappa^+(S^*(\overline{z})-I)\bigr)\biggl(h^-(z)-\frac{1}{2\pi}\chi_\varkappa^+\alpha\Gamma_0(A_{\varkappa}-zI)^{-1}u\biggr)=0,
\label{ultimate}
\end{equation}
where in order to pass from (\ref{penultimate}) to (\ref{ultimate}) we have used the fact that (see (\ref{eq:theta}))
\[
\chi_\varkappa^+\Theta_\varkappa(z)
=\bigl(I-\I\chi_\varkappa^+\alpha(B_{\I I}-M(z))^{-1}\alpha\bigr)\chi_\varkappa^+=\bigl(I+\chi_\varkappa^+(S^*(\overline{z})-I)\bigr)\chi_\varkappa^+,\ \ \ \ z\in{\mathbb C}_-.
\]
Hence, the expression (\ref{boundary_values}) vanishes, which concludes the proof.

The property
\begin{equation*}
\widetilde{N}_+^\varkappa=\bigl\{u\in\cH: \chi_\varkappa^-\alpha\Gamma_0(A_{\varkappa}-z
I)^{-1}u\in H^2_+(E)\bigr\}\,
\end{equation*}
is proved in a similar way.

\

\noindent\textbf{Proof of Proposition~\ref{lem:brothers-naboko}.}
%\begin{proof}

  Suppose that $z\in\complex_+$. If
  \begin{equation*}
    \int_\reals\frac{d\mu(s)}{s-z}\in H_+^2\,,
  \end{equation*}
then, by \cite[Thm.\,5.19]{MR1307384}, there exists a
function $f\in L^2(\reals)$ such that
\begin{equation*}
  \int_\reals\frac{f(s)ds-d\mu(s)}{s-z}=0\,.
\end{equation*}
Fix a $z_0\in\complex_+$, then
\begin{align*}
  0=&\int_\reals\frac{f(s)ds-d\mu(s)}{s-z}-
\int_\reals\frac{f(s)ds-d\mu(s)}{s-z_0}\\
=&(z-z_0)\int_\reals\frac{f(s)ds-d\mu(s)}{(s-z)(s-z_0)}\,.
\end{align*}
Thus, one has
\begin{equation*}
  \int_\reals\frac{1}{s-z}\frac{f(s)ds-d\mu(s)}{s-z_0}=0\,,\quad\text{ for
    all } z\in\complex_+\setminus\{z_0\}\,,
\end{equation*}
where now $(s-z_0)^{-1}(f(s)ds-d\mu(s))$
is a complex measure on $\reals$. Further, we invoke to the upper
half-plane counterpart of the theorem by F. and M. Riesz obtained by
applying the conformal mapping from the unit circle onto the upper
half plane \cite[Chap.\,2,
Sec.\,A]{MR1669574}. This theorem implies that
$(s-z_0)^{-1}(f(s)dt-d\mu(s))$ is
absolutely continuous with respect to the Lebesgue measure
and, therefore, the same applies to $d\mu(s)$.

The case of  $H_-^2$ is treated likewise.
%\end{proof}

\def\cprime{$'$} \def\lfhook#1{\setbox0=\hbox{#1}{\ooalign{\hidewidth
  \lower1.5ex\hbox{'}\hidewidth\crcr\unhbox0}}}


\begin{thebibliography}{10}

\bibitem{MR0206711}
V.~M.~Adamjan, D.~Z.~Arov.
\newblock Unitary couplings of semi-unitary operators. (Russian)
\newblock {\em Mat. Issled.}, 1(2):3--64, 1966;
\newblock English translation in Amer. Math Soc. Transl. Ser. 2, 95, 1970

\bibitem{AdamyanPavlov}
Adamyan, V. M.; Pavlov, B. S. Zero-radius potentials and M. G. Kre\u\i n's formula for generalized resolvents. (Russian) ; translated from {\em Zap. Nauchn. Sem. Leningrad. Otdel. Mat. Inst. Steklov. (LOMI)} 149 (1986), Issled. Line\u\i n. Teor. Funktsi\u\i. XV, 7--23, 186 {\em J. Soviet Math.} 42 (1986), no. 2, 1537--1550

%\bibitem{MR2289696}
%J.~Behrndt and M.~Langer.
%\newblock Boundary value problems for elliptic partial differential operators
%  on bounded domains.
%\newblock {\em J. Funct. Anal.}, 243(2):536--565, 2007.

%\bibitem{MR2386256}
%J.~Behrndt, M.~M. Malamud, and H.~Neidhardt.
%\newblock Scattering theory for open quantum systems with finite rank coupling.
%\newblock {\em Math. Phys. Anal. Geom.}, 10(4):313--358, 2007.


%\bibitem{MR3013208}
%G.~Berkolaiko and P.~Kuchment.
%\newblock {\em Introduction to quantum graphs}, volume 186 of {\em Mathematical
%  Surveys and Monographs}.
%\newblock American Mathematical Society, Providence, RI, 2013.

%\bibitem{MR0080271}
%M.~\v{S}.~Birman.
%\newblock On the theory of self-adjoint extensions of positive definite operators.
%\newblock {\em Math. Sb. N. S.}, 38(80):431--450, 1956.

\bibitem{Amrein}Amrein, W. O.; Pearson, D. B. 
\newblock $M$ operators: a generalisation of Weyl-Titchmarsh theory. 
\newblock {\em J. Comput. Appl. Math.} 171, no. 1-2: 1--26, 2004.



\bibitem{Birman_1963}
M.~\v{S}.~Birman
\newblock Existence conditions for wave operators. (Russian)
\newblock {\em Izv. Akad. Nauk SSSR Ser. Mat.}, 27: 883--906, 1963.

\bibitem{MR0139007}
M.~\v{S}.~Birman, M.~G.~Kre\u\i n.
\newblock On the theory of wave operators and scattering operators. (Russian)
{\em Dokl. Akad. Nauk SSSR} 144:475--478, 1962.

\bibitem{MR1192782}
M.~\v{S}.~Birman and M.~Z.~Solomjak.
\newblock {\em Spectral theory of selfadjoint operators in {H}ilbert space}.
\newblock Mathematics and its Applications (Soviet Series). D. Reidel
  Publishing Co., Dordrecht, 1987.
\newblock Translated from the 1980 Russian original by S. Khrushch{\"e}v and V.
  Peller.

\bibitem{MR0015185}
G.~Borg.
\newblock Eine Umkehrung der Sturm-Liouvilleschen Eigenwertaufgabe. Bestimmung der Differentialgleichung durch die Eigenwerte. (German)
\newblock Acta Math. 78:1--96, 1946.

\bibitem{MR0058063}
G.~Borg.
\newblock Uniqueness theorems in the spectral theory of
  {$y''+(\lambda-q(x))y=0$}.
\newblock In {\em Den 11te {S}kandinaviske {M}atematikerkongress, {T}rondheim,
  1949}, pages 276--287. Johan Grundt Tanums Forlag, Oslo, 1952.

\bibitem{Brodski}
M.~S.~Brodskij. \emph{Triangular and Jordan representations of linear
  operators}, Translations of Mathematical Monographs. Vol. 32. Providence, R.I.: American Mathematical Society, 1971.

\bibitem{MR2418300}
M.~Brown, M.~Marletta, S.~Naboko, and I.~Wood.
\newblock Boundary triples and {$M$}-functions for non-selfadjoint operators,
  with applications to elliptic {PDE}s and block operator matrices.
\newblock {\em J. Lond. Math. Soc. (2)}, 77(3):700--718, 2008.

\bibitem{BMNW2018}
M.~Brown, M.~Marletta, S.~Naboko, and I.~Wood.
\newblock The functional model for maximal dissipative operators: An approach in the spirit of operator knots, 29 pp., {\it arXiv:1804.08963.}

\bibitem{CEK}
K.~D.~Cherednichenko, Yu.~Yu.~Ershova, and A.~V.~Kiselev, 2019. Time-dispersive behaviour as a feature of critical contrast media,  {\it SIAM J. Appl. Math.} {\bf 79}(2), 690--715.
%K.~Cherednichenko, Y.~Ershova, A.~Kiselev. Time-dispersive behaviour as a feature of critical contrast media, 21pp., {\it arXiv:1803.09372}.

\bibitem{CEKN}
K.~Cherednichenko, Y.~Ershova, A.~Kiselev, and S.~Naboko, 2019. Unified approach to critical-contrast homogenisation with explicit links to time-dispersive media. {\it Trans. Moscow Math. Soc.} {\bf 80}(2), 295--342.
%K.~Cherednichenko, Y.~Ershova, A.~Kiselev, S.~Naboko. Unified approach to critical-contrast homogenisation with explicit links to time-dispersive media, 37 pp., {\it arXiv:1805.00884.}




\bibitem{CherednichenkoKiselevSilva}
K.~D.~Cherednichenko, A.~V.~Kiselev, L.~O.~Silva.
\newblock Functional model for extensions of symmetric operators and applications to scattering theory.
\newblock {\em Netw. and Heterog. Media,} {\bf 13}(2): 191--215, 2018.

\bibitem{MR0622619}
P.~Deift, E.~Trubowitz.
\newblock Inverse scattering on the line.
\newblock {\em Comm. Pure Appl. Math.} 32(2):121--251, 1979.

\bibitem{Derkach}
V.~Derkach.
\newblock Boundary triples, Weyl functions, and the Kre\u\i n formula.
\newblock {\em Operator Theory: Living Reference Work}, DOI 10.1007/978-3-0348-0692-3\_32-1
\newblock Springer Basel, 2015

%\bibitem{MR2486805}
%V.~Derkach, S.~Hassi, M.~Malamud, and H.~de~Snoo.
%\newblock Boundary relations and generalized resolvents of symmetric operators.
%\newblock {\em Russ. J. Math. Phys.}, 16(1):17--60, 2009.

%\bibitem{MR890193}
%V.~A. Derkach and M.~M. Malamud.
%\newblock On the {W}eyl function and {H}ermite operators with lacunae.
%\newblock {\em Dokl. Akad. Nauk SSSR}, 293(5):1041--1046, 1987.

\bibitem{MR1087947}
V.~A. Derkach and M.~M.~Malamud.
\newblock Generalized resolvents and the boundary value problems for
  {H}ermitian operators with gaps.
\newblock {\em J. Funct. Anal.}, 95(1):1--95, 1991.

%\bibitem{MR1318517}
%V.~A. Derkach and M.~M. Malamud.
%\newblock The extension theory of {H}ermitian operators and the moment problem.
%\newblock {\em J. Math. Sci.}, 73(2):141--242, 1995.

\bibitem{MR3484377}
Y.~Ershova, I.~I. Karpenko, and A.~V. Kiselev.
\newblock Isospectrality for graph {L}aplacians under the change of coupling at
  graph vertices.
\newblock {\em J. Spectr. Theory}, 6(1):43--66, 2016.

\bibitem{MR3430381}
Y.~Ershova, I.~I. Karpenko, and A.~V. Kiselev.
\newblock Isospectrality for graph {L}aplacians under the change of coupling at
  graph vertices: necessary and sufficient conditions.
\newblock {\em Mathematika}, 62(1):210--242, 2016.

%\bibitem{MR3404107}
%Y.~Y. Ershova, I.~I. Karpenko, and A.~V. Kiselev.
%\newblock On inverse topology problem for {L}aplace operators on graphs.
%\newblock {\em Carpathian Math. Publ.}, 6(2):230--236, 2014.

%\bibitem{MR1459512}
%P.~Exner.
%\newblock A duality between {S}chr\"odinger operators on graphs and certain
%  {J}acobi matrices.
%\newblock {\em Ann. Inst. H. Poincar\'e Phys. Th\'eor.}, 66(4):359--371, 1997.

%\bibitem{MR0110466}
%L.~D.~Faddeev.
%\newblock The inverse problem in the quantum theory of scattering. (Russian)
%\newblock {\em Uspehi Mat. Nauk} 14(4): 57--119, 1959.
%\newblock English translation: The inverse problem in the quantum theory of scattering. Research Report No EM-165, New York University, 1960.

%\bibitem{Faddeev_NYU}
%L.~D.~Faddeyev.
%\newblock {\em The inverse problem in the quantum theory of scattering.}
%\newblock Research Report No EM-165, New York University, 1960.

\bibitem{MR0149843}
L.~D.~Faddeyev.
\newblock The inverse problem in the quantum theory of scattering.
\newblock {\em J. Mathematical Phys.}, 4:72--104, 1963.

\bibitem{Faddeev_additional}
L.~D.~Faddeev
\newblock The inverse problem in the quantum theory of scattering. II. (Russian) Current problems in mathematics, Vol. 3 (Russian),
\newblock Akad. Nauk SSSR Vsesojuz. Inst. Nau\v{c}n. i Tehn. Informacii, Moscow, 93--180, 1974.
\newblock English translation in:{\em J. Sov. Math.}, 5: 334--396, 1976.

\bibitem{Friedrichs}
K.~O.~Friedrichs.
\newblock On the perturbation of continuous spectra.
\newblock{\em Communications on Appl. Math.} 1: 361--406, 1948.

\bibitem{MR0045281}
I.~M.~Gel'fand, B.~M.~Levitan.
\newblock On the determination of a differential equation from its spectral function. (Russian)
\newblock Izvestiya Akad. Nauk SSSR. Ser. Mat. 15:309--360, 1951.

\bibitem{Krein}  Gohberg, I. C., Krein, M. G.,
\emph{Introduction to the theory of linear nonself-adjoint operators},
Translations of Mathematical Monographs, Vol. 18. AMS, Providence, R.I.,
1969.

\bibitem{MR1294813}
M.~L.~Gorbachuk and V.~I.~Gorbachuk.
\newblock The theory of selfadjoint extensions of symmetric operators; entire
  operators and boundary value problems.
\newblock {\em Ukra\"\i n. Mat. Zh.}, 46(1-2):55--62, 1994.

\bibitem{MR1466698}
M.~L.~Gorbachuk and V.~I.~Gorbachuk.
\newblock {\em M. {G}. {K}rein's lectures on entire operators}, volume~97 of
  {\em Operator Theory: Advances and Applications}.
\newblock Birkh\"auser Verlag, Basel, 1997.

\bibitem{MR1154792}
V.~I. Gorbachuk and M.~L. Gorbachuk.
\newblock {\em Boundary value problems for operator differential equations},
  volume~48 of {\em Mathematics and its Applications (Soviet Series)}.
\newblock Kluwer Academic Publishers Group, Dordrecht, 1991.
\newblock Translated and revised from the 1984 Russian original.

\bibitem{MR0328627KK}
I.~Kac and M.~G.~Kre\u\i n.
\newblock $R$-functions--analytic functions mapping upper half-plane into itself.
\newblock {\em Amer. {M}ath. {S}oc. {T}ransl. {S}eries 2}, 103:1--18, 1974.

\bibitem{MR0407617}
T.~Kato.
\newblock {\em Perturbation theory for linear operators}.
\newblock Springer-Verlag, Berlin, Second edition, 1976.
\newblock Grundlehren der Mathematischen Wissenschaften, Band 132.

\bibitem{MR0385604}
T.~Kato and S.~T.~Kuroda.
\newblock The abstract theory of scattering.
\newblock {\em Rocky Mountain J. Math.}, 1(1):127--171, 1971.

\bibitem{Kiselev}
A.~V.~Kiselev.
\newblock Similarity problem for non-self-adjoint extensions of symmetric
  operators.
\newblock In {\em Methods of spectral analysis in mathematical physics}, volume
  186 of {\em Oper. Theory Adv. Appl.}, pages 267--283. Birkh\"auser Verlag,
  Basel, 2009.

\bibitem{MR0365218}
A.~N.~Ko{\v{c}}ube{\u\i}.
\newblock Extensions of symmetric operators and of symmetric binary relations.
\newblock {\em Mat. Zametki}, 17:41--48, 1975.

\bibitem{MR0592863}
A.~N.~Ko\v cube\u\i.
\newblock Characteristic functions of symmetric operators
and their extensions (in Russian).
\newblock {\em Izv. Akad. Nauk Arm. SSR Ser. Mat.}, 15(3):219--232, 1980


\bibitem{MR1669574}
P.~Koosis.
\newblock {\em Introduction to {$H_p$} spaces}, volume 115 of {\em Cambridge
  Tracts in Mathematics}.
\newblock Cambridge University Press, Cambridge, second edition, 1998.
\newblock With two appendices by V. P. Havin [Viktor Petrovich Khavin].

\bibitem{Kostrykin_Schrader}
V.~Kostrykin and R.~Schrader.
\newblock The inverse scattering problem for metric graphs and the traveling salesman problem.
\newblock {\em Preprint arXiv:math-ph/0603010}, 2006.

%\bibitem{MR1671833}
%V.~Kostrykin and R.~Schrader.
%\newblock Kirchhoff's rule for quantum wires.
%\newblock {\em J. Phys. A}, 32(4):595--630, 1999.

%\bibitem{MR0024575} M.~G.~Kre\u\i n.
%\newblock Theory of self-adjoint extensions of semi-bounded
%Hermitian operators and applications II. (Russian)
%\newblock {\em Mat. Sb. N. S.}, 21(63):365--404, 1947.

\bibitem{MR0039895}
M.~G.~Kre\u\i n.
\newblock Solution of the inverse Sturm-Liouville problem. (Russian)
\newblock {\em Doklady Akad. Nauk SSSR (N.S.)} 76:21--24, 1951.

\bibitem{MR0058072}
M.~G.~Kre\u\i n.
\newblock On the transfer function of a one-dimensional boundary problem of the second order. (Russian)
\newblock {\em Doklady Akad. Nauk SSSR (N.S.)} 88:405--408, 1953.

\bibitem{MR0078543}
M.~G.~Kre\u\i n.
\newblock On determination of the potential of a particle from its S-function. (Russian)
\newblock Dokl. Akad. Nauk SSSR (N.S.) 105:433--436, 1955.

\bibitem{MR0048704}
M.~G.~Kre\u\i n.
\newblock The fundamental propositions of the theory of representations of Hermitian operators
with deficiency index {$(m,m)$}.
\newblock {\em Ukrain. Mat. \v Zurnal}, 1(2):3--66, 1949.

%\bibitem{MR2600145}
%P.~Kurasov.
%\newblock Inverse problems for {A}haronov-{B}ohm rings.
%\newblock {\em Math. Proc. Cambridge Philos. Soc.}, 148(2):331--362, 2010.

\bibitem{MR0217440}
P.~D.~Lax and R.~S.~Phillips.
\newblock {\em Scattering theory}.
\newblock Pure and Applied Mathematics, Vol. 26. Academic Press, New
  York-London, 1967.

\bibitem{MR0032067}
N.~Levinson.
\newblock The inverse Sturm-Liouville problem.
\newblock {\em Mat. Tidsskr. B.} 1949:25--30, 1949.

\bibitem{Levinson}
N.~Levinson.
\newblock A certain explicit relation between the phase shift and scattering potential
\newblock {\em Phys. Rev.} 89:755-757, 1953

\bibitem{MR0020719}
M.~S. Livshitz.
\newblock On a certain class of linear operators in {H}ilbert space.
\newblock {\em Rec. Math. [Mat. Sbornik] N.S.}, 19(61):239--262, 1946.

\bibitem{Mak_Vas}
N.~G.~Makarov, V.~I.~Vasjunin. A model for noncontractions and stability of the continuous
spectrum.  {\it Lecture Notes in Math.}, 864:365--412, 1981.

\bibitem{MR0075402}
V.~A.~Mar\v{c}enko.
\newblock On reconstruction of the potential energy from phases of the scattered waves. (Russian)
\newblock {\em Dokl. Akad. Nauk SSSR (N.S.)} 104:695--698, 1955.

\bibitem{Marchenko_book}
V.~A.~Mar\v{c}enko.
\newblock  {\em Sturm-Liouville Operators and Applications (2nd ed.)}
\newblock Providence: American Mathematical Society, 2011.

%\bibitem{Marchenko_book_original}
%V.~A.~Marchenko.
%\newblock {\em Sturm-Liouville operators and applications. Operator Theory: Advances and Applications, 22}
%\newblock . Birkh\"auser Verlag, Basel, 1986.

\bibitem{MR0500225}
S.~N. Naboko.
\newblock Absolutely continuous spectrum of a nondissipative operator, and a
  functional model. {I}.
\newblock {\em Zap. Nau\v cn. Sem. Leningrad. Otdel Mat. Inst. Steklov.
  (LOMI)}, 65:90--102, 204--205, 1976.
\newblock Investigations on linear operators and the theory of functions, VII.

\bibitem{MR573902}
S.~N. Naboko.
\newblock Functional model of perturbation theory and its applications to scattering theory.
\newblock {\em Trudy Mat. Inst. Steklov.}, 147:86--114, 203, 1980.
\newblock Boundary Value Problems of Mathematical Physics, 10.

\bibitem{MR1036844}
S.~N.~Naboko.
\newblock Nontangential boundary values of operator {$R$}-functions in a half-plane.
\newblock {\em Algebra i Analiz}, 1(5):197--222, 1989.

\bibitem{MR1252228}
S.~N.~Naboko. On the conditions for existence of wave operators in the nonselfadjoint case. {\it Wave propagation. Scattering theory, Amer. Math. Soc. Transl. Ser. 2}, 157:127--149, Amer. Math. Soc., Providence, RI, 1993.

%\bibitem{MR0203530}
% B.~S.~Pavlov.
% \newblock On a non-selfadjoint Schrödinger operator. (Russian)
% \newblock {\em Probl. Math. Phys., No. I, Spectral Theory and Wave Processes}, pp. 102--132,
% \newblock Izdat. Leningrad. Univ., Leningrad, 1966.

\bibitem{MR0365199}
B.~S.~Pavlov.
\newblock Conditions for separation of the spectral components of a dissipative operator.
\newblock {\em Izv. Akad. Nauk SSSR Ser. Mat.}, 39:123--148, 240, 1975. English translation in: {\em Math. USSR Izvestija}, 9:113--137, 1975.

\bibitem{Drogobych}
B.~S.~Pavlov.
\newblock Diation theory and the spectral analysis of non-selfadjoint differential operators.
\newblock {\em Proc. 7th Winter School, Drogobych, 1974}, TsEMI, Moscow, 2--69, 1976. English translation: {\it Transl., II Ser., Am. Math. Soc} 115:103--142, 1981.

\bibitem{MR0510053}
B.~S.~Pavlov. Selfadjoint dilation of a dissipative Schr\"{o}dinger operator, and expansion in its eigenfunction. (Russian) {\it Mat. Sb. (N.S.)} 102(144): 511--536, 631, 1977.

%\bibitem{MR0385642}
%B.~S. Pavlov.
%\newblock Selfadjoint dilation of a dissipative {S}chr\"odinger operator, and
%  expansion in eigenfunctions.
%\newblock {\em Funkcional. Anal. i Prilo\v zen.}, 9(2):87--88, 1975.

\bibitem{MR0328674}
D.~B.~Pearson. Conditions for the existence of the generalized wave operators.
{\em J. Mathematical Phys.} 13:1490--1499, 1972.

\bibitem{MR529429}
M.~Reed and B.~Simon.
\newblock {\em Methods of modern mathematical physics. {III}}.
\newblock Academic Press [Harcourt Brace Jovanovich Publishers], New York,
  1979.
\newblock Scattering theory.

\bibitem{Romanov}
R.~Romanov. On the concept of absolutely continuous subspace for nonselfadjoint operators. {\it J. Operator Theory} 63(2):375--388, 2010.

\bibitem{MR822228}
M.~Rosenblum and J.~Rovnyak.
\newblock {\em Hardy classes and operator theory}.
\newblock Oxford Mathematical Monographs. The Clarendon Press Oxford University
  Press, New York, 1985.
\newblock Oxford Science Publications.

\bibitem{MR1307384}
M.~Rosenblum and J.~Rovnyak.
\newblock {\em Topics in {H}ardy classes and univalent functions}.
\newblock Birkh\"auser Advanced Texts: Basler Lehrb\"ucher. [Birkh\"auser
  Advanced Texts: Basel Textbooks]. Birkh\"auser Verlag, Basel, 1994.

\bibitem{Ryzh_ac_sing}
V.~Ryzhov. Absolutely continuous and singular subspaces of a nonselfadjoint operator. {\it J. Math. Sci.
(New York)} 87(5):3886--3911, 1997.

\bibitem{MR2330831}
V.~Ryzhov.
\newblock Functional model of a class of non-selfadjoint extensions of
  symmetric operators.
\newblock In {\em Operator theory, analysis and mathematical physics}, {\em Oper. Theory Adv. Appl.}, 174:117--158. Birkh\"auser, Basel, 2007.


  \bibitem{Ryzhov_closed}
 V.~Ryzhov, Functional model of a closed non-selfadjoint operator. {\it Integral Equations Operator
Theory} 60(4):539--571, 2008.

 \bibitem{Ryzh_spec}
 V.~Ryzhov.
 \newblock Spectral boundary value problems and their linear operators, 38 pp., {\it arXiv:0904.0276;} in the present volume.


\bibitem{Ryzhov_equipped}
V.~Ryzhov, Equipped absolutely continuous subspaces and the stationary construction of wave
operators in nonselfadjoint scattering theory. {\it J. Math. Sci. (New York)} 85(2):1849--1866, 1997.

\bibitem{MR2953553}
K.~Schm{\"u}dgen.
\newblock {\em Unbounded self-adjoint operators on {H}ilbert space}, Volume 265
  of {\em Graduate Texts in Mathematics}.
\newblock Springer, Dordrecht, 2012.

\bibitem{Calderon} I. M. Stein \emph{Singular integrals and differentiability properties of functions.} Princeton Math Ser. vol. 30, Princeton Univ. Press, Princeton, NJ, 1970.

\bibitem{Solomyak}
B.~M.~Solomyak. Scattering theory for almost unitary operators, and a functional model. {\it J. Soviet
Math.} 61(2):2002--2018, 1992.


\bibitem{MR2760647}
B.~Sz.-Nagy, C.~Foias, H.~Bercovici, and L.~K{\'e}rchy.
\newblock {\em Harmonic Analysis of Operators on {H}ilbert Space}.
\newblock Universitext. Springer, New York, Second enlarged edition, 2010.

\bibitem{Tikhonov}
A.~S.~Tikhonov. An absolutely continuous spectrum and a scattering theory for operators with
spectrum on a curve, {\it St. Petersburg Math. J.} 7 (1):169--184, 1996.


\bibitem{MR1503053} J.~von~Neumann.
\newblock  \"{U}ber adjungierte Funktionaloperatoren.
\newblock {\em Ann. Math.} 33(2):294--310, 1932

\bibitem{MR0066944}
J.~von Neumann.
\newblock {\em Mathematical foundations of quantum mechanics}.
\newblock Princeton University Press, Princeton, 1955.
\newblock Translated by Robert T. Beyer.

%\bibitem{MR0079235}
%M.~Rosenblum.
%\newblock On the operator equation $BX-XA=Q.$
%\newblock {\em Duke Math. J.}, 23:263--269, 1956.

\bibitem{MR0090028}
M.~Rosenblum.
\newblock Perturbation of the continuous spectrum and unitary equivalence.
\newblock {\em Pacific J. Math.}, 7:997--1010, 1957.

%%\bibitem{Tutte}
%%W.~T.~Tutte.
%%\emph{Graph theory. With a foreword by C. St. J. A. Nash-Williams.}
%%\newblock Encyclopedia of Mathematics and its Applications, 21. Addison-Wesley Publishing Company, Advanced Book Program, Reading, MA, 1984.
%xxi+333pp.

%\bibitem{MR0052655}
%M.~I.~Vi\v{s}ik.
%\newblock On general boundary problems for elliptic differential equations (Russian).
%\newblock {\em Trudy Moskov. Mat. Ob\v{s}\v{c}.}, 1: 187--246, 1952.

%\bibitem{Weder2015}
%R.~Weder.
%\newblock Scattering theory for the matrix Schr\"{o}dinger operator on the half line with general boundary conditions.
%\newblock J. Math. Phys. 56(9): 092103, 24 pp, 2015.

%\bibitem{Weder2016}
%R.~Weder.
%\newblock Trace formulas for the matrix Schr\"{o}dinger operator on the half-line with general boundary conditions.
%\newblock J. Math. Phys. 57(11): 112101, 11 pp, 2016.

\bibitem{MR566954}
J.~Weidmann.
\newblock {\em Linear operators in {H}ilbert spaces}, volume~68 of {\em
  Graduate Texts in Mathematics}.
\newblock Springer-Verlag, New York, 1980.
\newblock Translated from the German by Joseph Sz{\"u}cs.

\bibitem{MR1180965}
D.~R. Yafaev.
\newblock {\em Mathematical scattering theory}, Volume 105 of {\em Translations
  of Mathematical Monographs}.
\newblock American Mathematical Society, Providence, RI, 1992.
\newblock General theory, Translated from the Russian by J. R. Schulenberger.
%  MR0024575, MR0052655
\end{thebibliography}
\end{document}